%% file: main.tex
\documentclass[11pt]{article}
\usepackage[utf8]{inputenc}
\usepackage{amsmath,amsfonts,amssymb}
\usepackage{amsthm}
\usepackage{bbm}
\usepackage{bm}
\usepackage{latexsym}
\usepackage[noadjust]{cite}
\usepackage{graphicx}
\usepackage{hyperref}
\usepackage{xcolor}
\usepackage{dirtytalk}
\usepackage{mathtools}
\usepackage[margin=1in]{geometry}
\usepackage{thm-restate}
\usepackage{enumitem}
\usepackage{tikz-cd}
\usepackage[linesnumbered,ruled]{algorithm2e}

\usepackage{todonotes}

\newtheorem{theorem}{Theorem}[section]
\newtheorem{corollary}[theorem]{Corollary}
\newtheorem{lemma}[theorem]{Lemma}
\newtheorem{observation}[theorem]{Observation}

\newtheorem{claim}[theorem]{Claim}
\theoremstyle{definition}
\newtheorem{definition}[theorem]{Definition}
\newtheorem{remark}[theorem]{Remark}

\newtheorem{fact}[theorem]{Fact}

\usepackage[nameinlink]{cleveref}
\Crefname{claim}{Claim}{Claims}
\Crefname{observation}{Observation}{Observations}
\Crefname{algorithm}{Algorithm}{Algorithms}
\Crefname{fact}{Fact}{Facts}
\crefname{theorem}{Theorem}{Theorems}
\crefname{proposition}{Proposition}{Propositions}
\crefname{definition}{Definition}{Definition}
\crefname{lemma}{Lemma}{Lemma}
\crefname{corollary}{Corollary}{Corollaries}
\crefname{ineq}{inequality}{inequalities}
\Crefname{equation}{Equation}{Equations}

\input{macros}

\newif\ifdraft
\drafttrue

\newif\ifnames
\namestrue

\newcommand{\toc}[1]{\pagenumbering{roman}\setcounter{tocdepth}{#1}\tableofcontents\newpage\pagenumbering{arabic}}
\usepackage{comment}

\title{Beyond Worst Case Local Computation Algorithms}
\ifnames
\author{Amartya Shankha Biswas\thanks{Supported by NSF award CCF-2310818.}\\MIT\\\texttt{asbiswas@mit.edu} \and Ruidi Cao\\MIT\\\texttt{ruidicao@mit.edu}\and 
Cassandra Marcussen\thanks{Supported in part by an NDSEG fellowship, and by NSF Award 2152413 and a Simons Investigator Award to Madhu Sudan.}\\Harvard University\\\texttt{cmarcussen@g.harvard.edu} \and
Edward Pyne\thanks{Supported by a Jane Street Graduate Research Fellowship and NSF awards CCF-2310818 and CCF-2127597.}\\MIT\\\texttt{epyne@mit.edu} \and Ronitt Rubinfeld\thanks{Supported by NSF awards CCF-2006664, DMS-2022448
 and CCF-2310818.}\\MIT\\\texttt{ronitt@csail.mit.edu} 
 \and Asaf Shapira\thanks{Supported in part
by ERC Consolidator Grant 863438.}\\Tel Aviv University\\\texttt{asafico@tauex.tau.ac.il} \and Shlomo Tauber\thanks{Supported in part
by ERC Consolidator Grant 863438.}\\Tel Aviv University\\\texttt{shlomotauber@mail.tau.ac.il}
}
\else
\date{}
\fi

\begin{document}
\begin{titlepage}
\maketitle
\begin{abstract}
    We initiate the study of Local Computation Algorithms
    on average case inputs. In the Local Computation Algorithm (LCA) model, 
    we are given probe access to a huge graph, 
    and asked to answer membership queries about some combinatorial structure on the graph, answering each query with sublinear work. We define a natural model of average-case local computation algorithms.

    For instance, an LCA for the $k$-spanner problem gives access to a sparse subgraph $H\subseteq G$ that preserves distances up to a multiplicative factor of $k$. Our first result builds LCAs for this problem assuming the input graph is drawn from a variety of well-studied random graph models -- Preferential Attachment, Uniform Attachment, and \ER with a variety of parameters. Our spanners achieve size and stretch tradeoffs that are impossible to achieve for general graphs, 
    while having dramatically lower query complexity than known worst-case LCAs.

    Finally, we investigate the intersection of LCAs with Local Access Generators (LAGs). Local Access Generators provide efficient query access to a random object. We explore the natural problem of generating an \ER random graph \textit{together} with a combinatorial structure on it. We show that this combination can be easier to solve than focusing on each problem by itself, by building a fast, simple algorithm that provides access to an \ER random graph together with a maximal independent set.
\end{abstract}

\addtocounter{page}{-1}
\thispagestyle{empty}
\end{titlepage}

\toc{2}

\input{100-introduction}
\input{300-preliminaries}
\input{500-ER-spanners}

\input{600-ER-sparse}

\input{700-PA-spanners}
\input{800-UA-spanners}
\input{900-joint-sampling}

\section*{Acknowledgments}
Part of this research was conducted while a subset of the authors were visiting the Simons Institute program on Sublinear Algorithms.

\bibliographystyle{alpha}
\bibliography{ref}

\end{document}

%% file: macros.tex
\newcommand{\N}{\mathbb{N}}

\newcommand{\ra}{\rightarrow}
\newcommand{\la}{\leftarrow}

\newcommand{\eps}{\varepsilon}

\newcommand{\tO}{\widetilde{O}}

\DeclareMathOperator{\polylog}{polylog}
\DeclareMathOperator{\poly}{poly}
\DeclareMathOperator*{\E}{\mathbb{E}}

\newcommand{\cA}{\mathcal{A}}

\newcommand{\cC}{\mathcal{C}}

\newcommand{\cE}{\mathcal{E}}

\newcommand{\cG}{\mathcal{G}}

\newcommand{\cL}{\mathcal{L}}
\newcommand{\cM}{\mathcal{M}}

\DeclareMathOperator{\ID}{ID}

\newcommand{\diam}{\mathrm{diam}}

\newcommand{\ER}{Erd\H{o}s-R\'{e}nyi~}

\newcommand{\Gnp}{G(n,p)}

\newcommand{\Geom}{\mathrm{Geom}}
\newcommand{\Gpa}{G_{pa}(n,\mu)}
\newcommand{\Gua}{G_{ua}(n,\mu)}

\DeclareMathOperator{\MIS}{MIS}
\DeclareMathOperator{\edge}{Edge}

\newcommand{\ts}{\texorpdfstring}

\newcommand{\sm}{\mathrm{Smallest}}
\newcommand{\lr}{\mathrm{Largest}}
\newcommand{\admin}{\mathrm{Admin}}
\newcommand{\Gnps}{G(n,p^*)}
\newcommand{\ps}{p^*}

\newcommand{\nbr}{\texttt{Nbr}}

%% file: 100-introduction.tex
\section{Introduction and Our Results}

When computing on a very large object, it can be important to find fast algorithms which answer user queries to the solution, while neither considering the whole input, nor computing the full output solution.
In the local computation model~\cite{RTVX11,ARVX12}, we are given probe access to a large object, such as a graph, 
and receive queries about some combinatorial structure on the graph.
We desire Local Computation Algorithms (LCAs) that can quickly answer such queries while making very few probes to the graph. 
Moreover, we require the answers returned by the algorithm on different queries to be consistent with a fixed global structure. 
This consistency requirement is challenging since
typically, we require the LCA to be memoryless -- it does not store information about its previous answers. 
\newcommand{\vr}{\vec{r}}
\begin{definition}\label{def:LCA}
	A \textbf{Local Computation Algorithm (LCA)} for a problem $\Pi$ is an oracle $\cA$ with the following properties.
	$\cA$ is given probe access to input $G$, a sequence of random bits $\vr$
    and a local memory.
	For any query $q$ in a family of admissible queries to the output,
	$\cA$ must use only its oracle access to $G$ (which we refer to as probes to $G$), random bits $\vr$,
	and local memory to answer the query $q$.
	After answering the query, $\cA$ erases its local memory (including the query $q$ and its response). 
    Let $T_{\cA}(G,q)$ denote the expected (over the choice of $\vr$)  number of
    probes it takes for the LCA $\cA$ to answer query $q$ on input $G$, and
    set $T_{\cA}(G) = \max_q T_{\cA}(G,q)$.
    We say the LCA has probe complexity $T(n)$ if the maximum of $T_{\cA}(G)$ over all possible inputs $G$ parametrized by size $n$ is $T(n)$.
	All the responses to queries given by $\cA$ must be consistent
	with a single valid solution $X$ to the specified computation problem on input $G$. 
\end{definition}
There has been extensive work on fast LCAs for 
a variety of
natural problems.
For example, on bounded degree graphs, there are LCAs with 
polylogarithmic query complexity for maximal independent set (MIS)~\cite{ghaffari2016improved, LRY17, ghaffari2019sparsifying, Gha22}, maximal matching~\cite{mansour2013local, yoshida2009improved, LRY17, BRR23}, and $(\Delta+1)$ vertex coloring~\cite{even2014deterministic, feige2017probe, czumaj2018sublinear, chang2019complexity}. LCAs have found applications in well-studied algorithmic problems \cite{LCAapp1, lange2025local} (such as matching) and have contributed to breakthroughs in learning theory \cite{LCAapp2, lange2025agnostic, lange2025robust}.

For other problems, polylogarithmic complexities for LCAs are ruled out by lower bounds.
For example, given a graph $G$, for the task of providing local access to
a spanner of $G$,
the best known query complexities~\cite{LRR16, LL18, PRVY19, LRR20, ACLP23,BF24}
are $\tO(n^{2/3})$,
and known lower bounds imply that $\Omega(\sqrt{n})$ time is required even for constant degree graphs~\cite{LRR16}.

A natural question is whether we can build improved LCAs when we assume the input graph is drawn from some distribution, 
and ask the LCA to succeed with high probability over a random graph from 
this distribution.
This motivates our definition of an average-case LCA:
\begin{definition}
	We say that $\cA$ is an \textbf{average-case local computation algorithm}
 for a distribution over objects $\cG$  parametrized by size $n$ 
 for problem $\Pi$
 if, with probability $(1-1/n)$ over $G\la \cG$, $\cA^G$ (the algorithm when given probe access to input $G$) is an LCA.
 We say that the LCA $\cA$ has \textbf{average-case probe complexity} $T(n)$ if the expected probe complexity $T_{\cA}(G)$ over $G\la \cG$ is $T(n)$. We say the LCA has \textbf{worst-case probe complexity} $T(n)$ if the maximum probe complexity $T_{\cA}(G)$ over $G\la \cG$ is $T(n)$.
\end{definition}
Note the requirement that with high probability over the object $G$, the LCA succeeds for \textit{every} query on this object.

\begin{remark}
    Prior work~\cite{PA3,PA1,PA2} has studied \say{local information algorithms (LIAs)} for preferential attachment graphs, a well-studied average-case graph model. LIAs are sublinear algorithms that use local information to return a set of nodes possessing some property. Probes are allowed only to vertices directly neighboring the already explored set.
    Certain LIA algorithms imply LCAs for spanners on preferential attachment graphs, and we give a detailed comparison in~\Cref{app:PA}.
\end{remark} 

We now describe our results on average case LCAs for graph spanners. 
We then describe a new model of local access
generation which locally generates a random object together with a combinatorial
structure, and give our results for locally generating a random graph
together with a maximal independent set.

\subsection{Our Results: Spanner LCAs for Average-Case Graphs}
Our first set of results focus on the well-studied problem of LCAs for spanners~\cite{LRR16, LL18, LRR20, PRVY19, ACLP23, BF24}. We study LCAs for spanners over the Erd\H{o}s-R\'{e}nyi, Preferential Attachment, and Uniform Attachment random graph models.
\begin{definition}
	A {\bf $k$-spanner} of a graph $G$ is a subgraph $H\subseteq G$ such that distances are preserved up to a multiplicative factor of $k$, which we refer to as the \textbf{stretch}.
\end{definition} 
For general graphs, a spanner with size $O(n^{1+1/k})$ and stretch $(2k+1)$ can be constructed in linear time~\cite{BS03}. Moreover, conditional on Erdos' girth conjecture~\cite{ErdosGirthBound} this size-stretch tradeoff is tight.

An LCA for the spanner problem has probe access to $G$, and answers queries of the form \say{is $(u,v)\in H$?}. We desire to minimize the number of edges retained in $H$, the per-query work, and the stretch. The recent work of Arviv, Chung, Levi, and Pyne~\cite{ACLP23} (building off several prior works~\cite{LRR16, LL18, LRR20, PRVY19}) constructed LCAs for spanners of stretch $\polylog(n)$ and size $\tO(n)$ with query complexity $\tO(\Delta^2n^{2/3})$, where $\Delta$ is a bound on the maximum degree, and $3$-spanners of size $\tO(n^{3/2})$ with query complexity $\tO(\sqrt n)$. 

For general graphs there is a lower bound of $\Omega(\sqrt n)$ work per query~\cite{LRR16,PRVY19}, even for graphs of bounded degree. 

Thus, to obtain faster algorithms we \textit{must} make a ``beyond worst case" assumption.
For the random graph models we consider, our algorithms achieve a size-stretch tradeoff that is impossible to achieve for general graphs under Erdos' girth conjecture, while simultaneously achieving a query time that is impossible for LCAs for general graphs.

\subsubsection{Local Computation Algorithms for the \ER Model}
Recall that $\Gnp$ denotes the \ER graph model with edge probability $p$, where each edge $(u,v)$ for $u\neq v$ is present independently with probability $p$.

To motivate our results, we first overview two simple constructions that we will compare against. First, for a graph $G\la \Gnp$, if we keep each edge in $G$ with probability $p'/p$ for some $p'<p$, we effectively sample a graph $H\subseteq G$ that is itself distributed as $G(n,p')$. It is well known that for any $p'\ge p_0=(2+\eps)\log(n)/n$, this graph will be connected with high probability. Moreover, the graph will be an expander whp and hence will have diameter (and thus stretch, when considered as a spanner of $G$) of $O(\log n)$. It is immediate that we can implement an LCA that keeps each edge of $G$ with probability $\min\{p_0/p,1\}$\footnote{The LCA uses its random tape to answer future queries to the same edge in a consistent fashion.}, and we can summarize the resulting algorithm in the following:

\begin{observation}\label{obs:ER_spanner_trivial}
	There is an average-case LCA for $G\la \Gnp$ for every $p$ that whp provides access to an $O(\log n)$-spanner with $O(n\log n)$ edges. Moreover, the LCA has probe complexity $1$.
\end{observation}
However, such a construction cannot provide constant stretch with nearly-linear edges, nor linear edges with any stretch (as any $p'$ that results in a linear number of expected edges will result in a disconnected graph with high probability). 

Furthermore, it can be shown~\cite{ZamirPersonal} that for $p\ge p_0=(2+\eps)\log(n)/n$, we can likewise obtain a spanner by having each each vertex retain two random edges (which can be implemented by scanning down the adjacency list), giving an LCA with the following properties:
\begin{observation}[\cite{ZamirPersonal}]\label{obs:ER_spanner_Or}
    There is an average-case LCA for $G\la \Gnp$ for every $p\ge p_0$ that whp provides access to an $O(\log n)$-spanner with $2n$ edges. Moreover, the LCA has probe complexity $O(np)$.
\end{observation}
This algorithm improves the edge count of~\Cref{obs:ER_spanner_trivial}, but retains superconstant stretch (and has a slower query time).
We improve on both constructions, by obtaining ultra-sparse spanners (i.e. with $n+o(n)$ edges) and constant stretch. For dense graphs, our results are as follows:

\begin{restatable}{theorem}{ERspanner}\label{thm:ER_spanner}
	For every $np=n^{\delta}$, there is an average-case LCA for $G\la G(n,p)$ that whp gives access to a $(2/\delta+5)$-spanner $H$ with $n+o(n)$ edges. Moreover, the LCA has probe complexity $1$ where we have access to a sorted adjacency list in $G$, and $O\left(\min\left\{n^{\delta},n^{1-\delta}\log n\right\}\right)$ otherwise.
\end{restatable}
In particular, for highly dense \textit{and} highly sparse graphs, we obtain a runtime $n^\eps$ for small $\eps$, beating the worst-case lower bound of $\Omega(\sqrt n)$.

Our last result in the \ER model focuses on sparse input graphs 
(for instance, those with $np=n^{o(1)}$.)
Here we consider the relaxed problem of producing a sparse spanning subgraph (LCAs for which have been studied before~\cite{ Expansion,LL18, PRVY19,LRR20, BF24}), where we do not bound the stretch. We note that known lower bounds~\cite{Expansion,PRVY19} imply a $\sqrt{n}$ query lower bound even for this problem on sparse graphs.

We are able to obtain an ultra-sparse connectivity-preserving subgraph for all edge probabilities greater than $p^*=7\log(n)/n$, only a constant factor above the connectivity threshold. Moreover, we achieve query time $\tO(\Delta)$, where $\Delta=np$ is the expected degree of the graph. 

\begin{restatable}{theorem}{ERSSS}\label{thm:ER_SSS}
    There is an average-case LCA for $G\la G(n,p))$ for every $p \ge 7\log n / n$ that w.h.p provides access to a sparse connected subgraph $H\subseteq G$, such that $H$ has $n+o(n)$ edges. Moreover, the LCA has probe complexity $O(\Delta\polylog(n))$.
\end{restatable}

\subsubsection{Local Computation Algorithms for the Preferential and Uniform Attachment Models}\label{subsec:PA}

Next, we construct spanner LCAs for preferential and uniform attachment graphs with a sufficiently high degree parameter. In the preferential attachment model (formally defined in~\Cref{def:PA}), the graph is constructed by sequentially inserting $n$ vertices. For each new vertex $v_i$, $\mu=\mu(n)$ edges are added to the graph; these may either be self loops or edges from $v_i$ to existing vertices $(v_1,\ldots,v_{i-1})$, where an edge is added to $v_j$ with probability proportional to the degree of $j$. (Afterwards, the vertices are permuted randomly, so that the algorithm cannot use the IDs to determine insertion order).
Such a model captures the property that high-degree nodes tend to accumulate additional connections.
The generation of preferential attachment graphs (and variants of it) have been extensively studied~\cite{BAgen1, BAgen2, BAgen3, BAgen4, BAgen5, BAgen6}.

Our spanner LCA for preferential attachment graphs of sufficiently high degree constructs a low stretch spanning tree, a stronger object than a spanning subgraph. A low-stretch spanning tree $H\subseteq G$ is a spanner with exactly $n-1$ edges, the minimum required even to preserve connectivity. 
\begin{restatable}{theorem}{PASpanner}
	\label{thm:PA_Spanner}
	For every preferential attachment process with parameter $\mu>c_{\mu}\log (n)$ for a global constant $c_{\mu}$, there is an average-case LCA for $G\la \Gpa$ that w.h.p gives access to an $O(\log n)$-spanner $H\subseteq G$, and moreover $H$ contains $n-1$ edges. On query $(u,v)$ the LCA has time complexity $O(d_u+d_v)$, which is $O(\mu \sqrt{n})$ in the worst-case and $O(\mu \log^3 n)$ in expectation (over all possible queries).
\end{restatable}

We note that techniques from previous work on local information algorithms \cite{PA2} can be applied to obtain an LCA giving access to a $\Tilde{O}(\log n)$-spanner with $\Tilde{O}(n)$ edges with a query time that is $O(\mu \sqrt{n})$ in the worst case and $O(\mu \polylog(n))$ in expectation for any $\mu$ (see \Cref{app:PA}). We focus on the setting where $\mu$ is sufficiently large and the resulting graph is therefore not already sparse. Our LCA in this setting achieves improved bounds for the sparsity and query time.

We also construct spanner LCAs for uniform attachment graphs of sufficiently high degree. In the uniform attachment model (formally defined in \Cref{def:UA}), $n$ vertices are also inserted sequentially. At each time step, a new vertex $v_i$ joins the graph and $\mu = \mu(n)$ edges are added from $v_i$ to existing vertices $(v_1, \ldots, v_{i-1})$, which are each chosen independently and uniformly at random. We assume again that the insertion order of vertices is unknown to the LCA. Uniform attachment graphs \cite{tapia1967generation} are standard models of random circuits and randomly evolving networks with applications including the modeling of networks, physical processes, and the spread of contamination among organisms, as noted in \cite{zhang2015number}.

\begin{restatable}{theorem}{UASpanner}
	\label{thm:UA_Spanner}
	For every uniform attachment process with parameter $\mu>c_{\mu}\log^2(n)$ for a global constant $c_{\mu}$, there is an average-case LCA for $G\la \Gua$ that w.h.p gives access to an $O(\log n)$-spanner $H\subseteq G$, and moreover $H$ contains $n + c$ edges, for some constant $c$ independent of $n$. 
 
    Moreover, let $D := \mu \cdot (H_{n-1} - H_6) + \mu/2$ where $H_n$ denotes the $n$-th Harmonic number. On query $(u,v)$, if $\min\{d_u, d_v\} > D$, the time complexity is $O(1)$. Otherwise, the time complexity is $O(d_u + d_v)$ which is $O(\mu \log n)$ in the worst-case and $O(\mu)$ in expectation (over all possible queries).
\end{restatable}

\subsection{Our Results: Joint Sampling of \ER Graphs and Maximal Independent Sets}
A natural topic relating to local algorithms and random graphs is to sample the random graph itself in a local fashion, rather than assuming we have probe access to one that already exists. Several recent works~\cite{GGN03,AN08,BRY20,BPR22,MSW22, ELMR} studied exactly this question, under the label of Local Access Generators (LAGs).
These algorithms provide efficient query access to a random instance of some structure.  
\begin{definition}\label{def:LAG}
	A \textbf{Local Access Generator} (LAG) of a random object $G$ sampled from a distribution $\cG$, is an oracle that provides access to $G$ by answering various types of
	\emph{supported queries}, given a sequence of random bits $\vec r$. We say the LAG is \textbf{memoryless} if it does not store its answers to prior queries. We require that (fixing a random tape) the responses of the local-access generator to all queries must be consistent with a single object $G$. Moreover, the distribution $\cG'$ sampled by the LAG must be within $n^{-c}$ from $\cG$ in TV distance, for any desired constant $c$.
\end{definition}
As in the case of LCAs, we desire Local Access Generators to be as efficient as possible per query. We also strongly desire the LAG to be memoryless (a requirement in the setting of LCAs,
but not always achieved for LAGs), and our result achieves this goal.

Given the two lines of work (local computation algorithms and local access generators), we ask if they can be unified. Rather than solving both problems independently, build an algorithm which provides access to a random graph $G\la \cG$ \textit{together} with a combinatorial structure $M$ on that graph. By jointly solving both problems, one could hope to exploit the ability for the local access generator and local computation algorithm to coordinate. 

Prior work has studied exactly this question in the setting of polynomial time algorithms. Work of Bach~\cite{Bach} showed that one could generate random numbers \textit{together} with their factorization, whereas factoring numbers that have been generated \say{in advance} is widely considered to be hard.

We show that such an approach is also fruitful in the setting of LCAs. We again focus on the dense \ER model, and this time on the extensively studied problem~\cite{ghaffari2016improved, LRY17, ghaffari2019sparsifying, Gha22} of Maximal Independent Set (MIS). 
The frontier result of Ghaffari~\cite{Gha22} provides an LCA for MIS with per-query runtime $\poly(\Delta\log n)$, and a local sampling implementation of dense \ER graphs is straightforward. However, composing these algorithms does not give a sublinear runtime. Our result achieves runtime $\polylog(n)$ for $p\ge 1/\polylog(n)$ per query, both for queries to the random graph and to its MIS:
\begin{restatable}{theorem}{ERMIS}\label{thm:ER_MIS}
	There is a memoryless Local Access Generator $\cA$ for $(G,M)$, where $G\la G(n,p)$ and $M\subseteq [n]$ is an MIS in $G$. Moreover, the per-query complexity of $\cA$ is $\polylog(n)/p$ with high probability. 
\end{restatable}

\section{Proof Overviews}

\subsection{Spanners for \ER}
Next, we overview our proofs. For our spanner results, we first give a \say{global} description of the connectivity condition, then describe how we implement this condition in a local fashion.

\ERspanner*
Our connectivity rule is as follows. We designate a sublinear-size set of vertices in $G$ as \textbf{centers}, which we denote $\cC$. We then retain in $H$ all edges between centers. Finally, every non-center vertex adds the first edge from itself to $\cC$ in $H$.

By choosing the size of $\cC$ appropriately, we ensure that the following three conditions hold with high probability: there are $o(n)$ intra-center edges (enforced by choosing the size so that  $|\cC|^2p=o(n)$), every non-center vertex has an edge to the center with high probability (enforced by choosing the size so that $|\cC|p=\Omega(\log n)$), and the center has constant diameter (which follows as the center is itself distributed as $G(|\cC|,p)$).

To implement this connectivity rule as an LCA, we break into the sparse case (where a non-center vertex simply queries its entire adjacency list and chooses the least-ranked edge to keep) or the dense case (where a non-center vertex queries the adjacency matrix until it finds its first edge). If we additionally assume that the adjacency list of each vertex is sorted in ascending order, we can perform this check in constant time.

\subsection{Sparse Connected Subgraphs for \ER}
Recall we are given probe access to $G\la \Gnp$ and wish to provide local access to a sparse connected subgraph $H\subseteq G$ with very few edges. Here we focus on the case where the input graph is itself somewhat sparse.
Without essential loss of generality, we assume the edge probability is exactly $p^* = 7\log n/n$ (as otherwise we can use the idea of~\Cref{obs:ER_spanner_trivial} to subsample as a first step).

We first describe the LCA as a $4$-round distributed algorithm, then use the approach of \cite{PR07} to show that we can implement it as an LCA with per-query work $O(\Delta \polylog n)$ (where we again use pre-sparsification to lower the probe complexity). For the formal proof, see~\Cref{sec:ER-SSS}.

First, we assume that all vertices have distinct indices drawn from some universe. Let $\Gamma(v)$ be the neighborhood of $v$ in $G$, and let $\sm(v)$ be the smallest index vertex in $\Gamma(v)$. First, if $\sm(v)<v$, we keep the edge $(v,\sm(v))$ in $H$, and broadcast to all other neighbors that we made this choice. Otherwise, if $v<\sm(v)$, we call $v$ a \textbf{candidate leader}. If $v$ is a candidate leader and receives at least one broadcast that it is \textit{not} being selected (which occurs if and only if $v$ does not have the least index in its two-hop neighborhood), we connect $v$ to the neighbor which allows it to reach the smallest $2$-hop neighbor.

After this connectivity rule, which has a simple two-round distributed algorithm, we call $v$ a \textbf{leader} if it has not added any out edges. This occurs if and only if it has the smallest index in its two-hop neighborhood.
Next, each leader retains an edge to its \textit{highest} index neighbor, which we call its \textbf{administrator}. Finally, each administrator keeps its entire neighborhood.

\textbf{Connectivity.} We define a set of events $\cE$ that partition the space of possible graphs, and denote a subset of events $\cE_G\subset \cE$ as \textbf{good}. We first show that a random graph lies in a good event with high probability. Next, we show that for \textit{every} good event $E\in \cE_G$, sampling a random $G$ that satisfies $E$ results in a graph that the algorithm succeeds on (in fact, we prove this with high probability over $G$).
    
Each such event specifies the presence or absence of a subset of edges in the graph. At a high level, these specifications capture the view of the algorithm up to the point that the leader vertices select their administrators. We define good events as those in which all administrators have many bits of entropy remaining in their neighborhoods, which allows us to argue they maintain connectivity with high probability. 

\textbf{Subgraph Size.} It is easy to see that each edge keeps at most one edge to its lowest index neighbor, and each leader candidate that is not a leader keeps at most one edge, so it suffices to show that the number of edges added in the final phase (when administrator vertices add their entire edge set) is sublinear in $n$. To do this, we show that the number of administrators is $O(n/\log^2 n)$, which itself follows from the fact that each leader has minimal rank in its $2$-hop neighborhood. Then as the degree of the graph is $O(\log n)$, we obtain a bound of $o(n)$ edges added in the final phase.

\textbf{Local Implementation.} One can see that the algorithm constitutes a $4$-round distributed algorithm, and hence can be implemented in per-query work $O(\Delta^4)$ via the reduction of Parnas and Ron~\cite{PR07}. However, we note that we can first subsample the graph $G$ to have edge probability $p^*=7\log n$ (which we do in a global fashion using the random tape of the LCA). By~\Cref{obs:ER_spanner_trivial}, this produces whp a connected subgraph of $G$ that is itself distributed $\Gnps$. Subsequently, in each distributed round we explore only the neighbors of $v$ that are retained in the subsampled graph, resulting in total work $O(\Delta \log^4 n)$.

\subsection{Spanners for Preferential and Uniform Attachment}\label{sec:PA-overview}
We highlight the main ideas behind the proofs of \Cref{thm:PA_Spanner} and \Cref{thm:UA_Spanner}. We are interested in the case where the degree parameter $\mu$ is sufficiently high (and thus the total number of edges $n \cdot \mu$ is high and constructing spanners is compelling). When $\mu$ is large, both preferential and uniform attachment graphs are well-structured, in the sense that higher degree vertices typically have an earlier arrival time. At a high level, we can leverage this structure by having each vertex keep an edge to its highest-degree neighbor, thus generally having the spanner keep paths from vertices to the earliest-added vertices in the process.

For preferential attachment graphs (\Cref{thm:PA_Spanner}), we use the following algorithm to build low-stretch spanning trees. Because preferential attachment graphs are multigraphs, the algorithm's input specifies the two nodes adjacent to the edge as well as the lexicographic indexing of the edge. (See \Cref{remark:multigraph}.)
\\\\
\textbf{Algorithm:} On query $(u,v, i)$: \\
Check if $v$ is the highest-degree neighbor of $u$ or vice versa. If so, keep the edge if it is the lexicographically first edge between $u$ and $v$ (i.e. $i = 1$). Otherwise, discard the edge.
\\\\
The proof is a direct consequence of a structural result about preferential attachment graphs with $\mu(n)\ge c_\mu \log n$: With high probability, every vertex $v$ that is not the highest-degree vertex is either a neighbor of the highest-degree vertex, or $v$ has a neighbor $u$ such that $d_u > 2\cdot d_v$.

To show this structural result, we prove that there is a global constant $c$ such that with high probability, every $v_{i}$ with arrival time $i\ge c$ has a neighbor with degree at least $2d_{v_i}$, and every $v_{i},v_j$ for arrival times $i,j\le c$ are connected.
The second item is a simple consequence of our choice of $\mu$, and the first follows from a tail bound for the degrees of preferential attachment graphs established by~\cite{DKR18}.
Given this result, connectivity is direct, and a simple potential argument (that the degree cannot increase by a factor of $2$ more than $\log(m)$ times, where $m$ is the total degree of the graph) establishes a stretch bound of $O(\log(d_{\max}))=O(\log n)$. Furthermore, we show that the worst-case runtime is $O(\mu \sqrt{n})$ and the average-case runtime over all possible queries is $\mu \cdot \polylog(n)$. To see the size bound, note that every vertex that is not of globally highest degree adds exactly one edge to $H$.
For a formal proof, see \Cref{app:PA}.

Moving to the setting of uniform attachment graphs (\Cref{thm:UA_Spanner}), we utilize a similar algorithm that keeps an edge $(u, v)$ if $u$ is the highest-degree neighbor of $u$ or vice-versa. We additionally keep an edge between vertices that have a degree above some threshold (to ensure that edges are kept between the earliest-added vertices, whose degrees will all be similar). The similarity between the algorithms for preferential and uniform attachment demonstrates the robustness of our approach to different settings of randomly growing graph processes. We present the algorithm in \Cref{sec:UA_spanners}. 

We prove that every vertex $v_i$ is connected to at least one vertex added sufficiently earlier in the process, and such a vertex can be identified because it will have a higher degree than neighbors added later. The LCA keeps an edge to such a vertex. More formally, we define subsets of vertices $C_1 \supset C_2 \supset \dots \supset C_M$ for some $M$ that is $O(\log n)$; each $C_m$ is defined as the first $|C_m|$ vertices added to the process. We show that, with high probability, all vertices in $[n] \setminus C_1$ are connected to a vertex in $C_1$, and similarly all vertices whose time of arrival $i$ satisfies $|C_{m+1}| < i < e^2 \cdot |C_m|$ are connected to a vertex in $C_{m + 1}$, for all $m \in \{1, 2, \dots, M\}$. The $e^2$ arises because degrees are not precise indicators of times of arrival (see \Cref{sec:ua-structural}). 

Finally, \Cref{thm:UA_Spanner} is proven by arguing that, when 
each vertex keeps an edge to its highest-degree neighbor, the LCA keeps paths of length $O(\log n)$ from each vertex to the vertices added in the first $t$ steps, for a small constant $t$. All adjacencies between the vertices added in the first  $t$ steps are preserved, resulting in a low-stretch sparse spanner. For the formal proofs, see \Cref{sec:UA_spanners}.

\subsection{Joint Sampling of \ER Graphs and Maximal Independent Sets}
We describe the sampling algorithm in a global fashion, and then describe how to implement it via an LCA. We gradually grow the MIS $M$, by instantiating $M=(1)$ where $1$ is the first vertex, and sequentially determining the smallest vertex that is not connected to $M$, and add it to $M$, continuing until we exhaust the vertex set. For a fixed $M=(v_1,\ldots,v_t)$, each subsequent vertex $u>v_t$ is not connected to any element of $M$ in $G(n,p)$ with probability $(1-p)^{|M|}$. For a fixed $u$, we can determine if $(v_i,u)$ is in $G$ for every $i$ by simply sampling the edge, and thus determine if $u$ should be added to $M$. However, this procedure would take linear time to determine the MIS. Instead, we use local sampling of the Geometric distribution to find the next element of $M$. Once we sample $r \sim \Geom(1-p)$, we let the next element of $M$ be $u=v_t+r$. For all $u' \in (v,u)$, this virtually conditions on the event that at least one of $(u,v_i)$ is present in the graph for $i\le t$. 

We can determine the entire $M$ in this fashion in time $\polylog(n/p)$. Then we can answer queries as follows. On an $\MIS(a)$ query, we recompute $M$ (using the same random bits) and answer whether $a\in M$. To answer edge queries, we must be careful to not contradict the queries made to the MIS. To achieve this, on receiving the query $\edge((a,b))$ we first re-determine $M$ (using the same random bits as before), and then answer the query as follows:
\begin{itemize}
	\item If $a,b\in M$, we say the edge is not present
	\item If $a,b\notin M$, we use independent random bits to sample if the edge is present
	\item If $a\in M, b\notin M$, we work as follows. First, let $v_i<b<v_{i+1}$ be the elements of $M$ that bracket $b$. Note that by the setup of the sampling procedure, we have conditioned on the event that there is at least one edge $(v_j,b)\in G$ for $j\le i$. All other edges to centers have not been determined by the sampling process, so if $a=v_j$ for $j>i$, we use independent random bits to sample if the edge is present. Otherwise, we must determine the set of edges
    \[
        (v_1,b),\ldots,(v_i,b)
    \]
    subject to the constraint that at least one such edge is present.
    To do this, we sample all edges in this set independently at random, and reject and retry if no edges are added. This will determine the edge set after $O(\log n)$ retries with high probability. Once we do this, 
	we can answer the query on $(a,b)$.
\end{itemize}

%% file: 300-preliminaries.tex
\section{Preliminaries}\label{sec:prelims}
We first define our access model and define some required lemmas.
\paragraph{Access Model.}
We assume that an LCA has access to a graph $G=([n],E)$ via the following probes:
\begin{itemize}
    \item \texttt{Exists}$(u, v)$ returns true/false based on whether the edge $(u, v)\in E$
    \item \texttt{Deg}$(v)$ returns the degree of vertex $v$ in $G$
    \item \nbr$(v, i)$ returns the $i^{th}$ neighbor of $v$ from the adjacency list
    if $i\le$ \texttt{Degree}$(v)$, and $\bot$ otherwise
\end{itemize}

\begin{remark}[Remark on multigraphs]\label{remark:multigraph}
    Certain random models we consider will produce multigraphs with high probability. Without loss of generality, we assume that if $\nbr(v,i)=\nbr(v,j)$ for $i<j$ (i.e. there is a multi-edge), the second query returns $\perp$ (alternatively, our algorithm only retains the lexicographically first edge).
\end{remark}

\paragraph{Graph Models}
We formally define the \ER graph model. Note that in all proofs, $\Gamma(v)=\Gamma_G(v)$ refers to the neighbors of $v$ in the original graph $G$ that the LCA has probe access to. We will always denote the original graph as $G$, and subgraphs or spanners as $H$.

\begin{definition}[\ER graphs]\label{def:ER}
    For a function $p=p(n)$, we say $G\la \Gnp$ is an \textbf{\ER random graph} if it is constructed as follows. For every pair of vertices $u,v\in [n]$ with $u\neq v$, we add the edge $(u,v)$ with probability $p$, independently.
\end{definition}

\paragraph{Sampling Algorithms.}
We recall an efficient algorithm for sampling from the geometric distribution with parameter $\lambda$, such as the one used in \cite{BRY20}.
\begin{lemma}\label{lem:GeomSamp}
    There is a randomized algorithm that, given $n$ and $\lambda>0$, samples from $\Geom(\lambda)$ in $\polylog(n/\lambda)$ time with high probability.
\end{lemma}

%% file: 500-ER-spanners.tex
\section{Spanners for \ER Graphs}\label{sec:ER_spanners}
In this section,  
we present our results on spanners for \ER graphs.
We show that we can achieve a provably superior size-stretch than what would be obtained by naive subsampling.

\ERspanner*
We first note that we require a bound on the diameter of random graphs:
\begin{theorem}[Theorem 7.1~\cite{FK15}]\label{thm:ER_diam}
    There is a constant $c$ such that for every $d\in \N$ and $p=p(n)$, if 
    \[
    (pn)^d \ge n\log(n^2/c)
    \]
    then the diameter of $G\la \Gnp$ is at most $d+1$ with high probability.
\end{theorem}
At a high level, our proof designates the first $T$ vertices in the graph as centers, for appropriate chosen $T$. We keep all edges between centers, which results in a sublinear number of edges, and for each non-center vertex keep the lowest ranked edge into the center cluster.

\begin{proof}[Proof of Theorem ~\ref{thm:ER_spanner}]
    We give a global description of the process and then describe the (simple) local implementation of the process. Let
    \[ T = n^{1-\delta/2-\delta^2/8}.
    \]
    Recall that we are given probe access to $G\la \Gnp$.
    For every vertex with $\ID(v)\le T$, let $v$ be a center vertex. Denote the set of center vertices as $C\subseteq V$. 
    When queried on an edge $(u,v)\in G$,
    we let our connectivity rule be as follows. 
    We keep all edges in $G$ internal to the center, 
    and for each non-center vertex $v$ keep only the edge in $G$ that connects $v$ to the lowest-ranked element of the center. 
    More formally:
    \[
    (u,v)\in H\iff \{u\in C,v\in C\} \textrm{ OR } \{u\not\in C, v\in C, v=\min\{\Gamma(u)\cap C\}\}
    \]
    We now prove that the edge counts, diameter, and query times are as claimed. For both, observe that the center graph $\cC=H_{C\times C}$ is itself distributed as $G(|T|,p)$. 

    \paragraph{Sparsity.} Observe that $|E(H)|\le n+|E(\cC)|$ as each non-center vertex contains at most one edge. Next, note that for every pair $u,v\in C$ we have that the event that $(u,v)\in G$ occurs with probability $p$ and is independent. We have that the expected edge count is
    \[
    \E[|E(\cC)|] = \binom{|C|}{2} p \le n^{2-\delta-\delta^2/4}n^{\delta-1}=o(n).
    \]
    Moreover, we can apply a Chernoff bound and conclude that the edge count in $\cC$ is at most $o(n)$ with overwhelming probability.

    \paragraph{Diameter.} We first claim that $\diam(H)\le 2+\diam(\cC)$. For every vertex $v\notin C$, we claim that $|\Gamma(v)\cap C|\ge 1$ with overwhelming probability, and hence $v$ will be connected to a center vertex and so the above bound holds by considering paths through the center. We have
    \[
    \E[|\Gamma(v)\cap C|] = |C|p=n^{\delta/2-\delta^2/8}
    \]
    and again applying the Chernoff bound, we have that this set is of nonzero size with probability at least $1-n^{-5}$.
    
    Then by~\Cref{thm:ER_diam} with $d=\lceil 2/\delta\rceil+1$, $p=p$, $n=|C|$, and verifying that 
    \[
    (|C|p)^d = (n^{\delta/2-\delta^2/8})^{\lceil 2/\delta \rceil +1} > n > |C|\log(|C|^2/c)
    \]
    we obtain that the center graph has diameter at most $2/\delta+3$, so we are done.

    \paragraph{Local Implementation.} For every query $(u,v)$, we can determine whether $u,v$ are in the center by querying their IDs. If $u\in C$ and $v\in C$, we are done, and likewise if $u\notin C,v\notin C$. Finally, suppose $v\in C$ and $u\notin C$. 
        
    If the adjacency list that we have probe access to is sorted, we probe \nbr$(u,1)$ to obtain the first neighbor of $u$ (which is the edge we keep to the center), and if this neighbor is $v$ we return $(u,v)\in H$, and otherwise return $(u,v)\notin H$. If the adjacency list is not sorted, we find if $(u,v)$ is the least ranked edge from $u$ to $C$ in one of two ways. If $n^\delta \le n^{1-\delta}$, we enumerate the neighborhood of $u$ using \nbr$(u,i)$ queries and use this information to decide. Otherwise, we query \[
    \text{\texttt{Exists}}(u,1),\ldots,\text{\texttt{Exists}}(u,k)\]
    until we find an edge, which occurs after $O(\log(n)/p)$ probes with high probability.
\end{proof}

%% file: 600-ER-sparse.tex
\addtocontents{toc}{\protect\setcounter{tocdepth}{1}}

\section{Sparse Connected Subgraphs for \ER}\label{sec:ER-SSS}

We first give a distributed algorithm when $p$ is exactly equal to $p^*=7\log(n)/n$, and then extend this into a full LCA for larger $p$ (\Cref{thm:ER_SSS}).
\begin{restatable}{theorem}{ER_SSS_pstar}
\label{thm:ER_SSS_pstar}
    There is a 4-round distributed algorithm providing access to a subgraph $H\subseteq G\la G(n,p^*)$ such that with high probability, $H$ is connected and has at most $(1 + c/\log(n))n$ edges.
\end{restatable}
We first use~\Cref{thm:ER_SSS_pstar} to prove the main result. The algorithm of~\Cref{thm:ER_SSS_pstar} is described formally as Algorithm~\ref{alg:ER_SSS_pstar}.

Using this distributed algorithm and the subsampling technique, we obtain the full result.
\ERSSS*
\begin{proof}
    For every edge $(u,v)\in G$, we retain the edge independently with probability $p^*/p$ into the subgraph $G'$, and simulate Algorithm~\ref{alg:ER_SSS_pstar} using the reduction of Parnas and Ron~\cite{PR07} on $G'$. First, note that $G'$ is connected w.h.p. (as $p^*$ is above the connectivity threshold), and moreover is distributed as $\Gnps$ (over the randomness of $G$ and the sparsification step).

    The Parnas-Ron reduction takes a $k$-round distributed algorithm, and simulates an LCA for a query on vertex $v$ by exploring the $k$-neighborhood ov vertex $v$ (which has size $O(\Delta^k)$),
    and simulating the local execution of the distributed algorithm within the $k$-neighborhood.
    Considering the $k$ frontiers of BFS starting at the vertex $v$,
    notice that we can simulate the $i^{th}$ round of the distributed algorithm on every vertex within the $i^{th}$ frontier by performing this process starting at $i=1$ and proceeding for successive rounds $i$.
    
    We use a modified version of the Parnas-Ron reduction to take into account the nature of the pre-sparsification step.
    Instead of exploring the entire $O(\Delta^k)$ sized neighborhood, we can instead prune away the edges which are not included in the pre-sparsification step at each frontier.
    Since the max degree after sparsification is $O(\log n)$ w.h.p., the total number of edges explored in order to find the sparsified neighborhood is $O(\Delta \log^{k-1} n)$.
    This is obtained by performing a BFS traversal and pruning away non-sparsified edges at each frontier.
    Once the $O(\log^k n)$ sized neighborhood is discovered, we can then simulate the distributed algorithm in $O(\polylog (n))$ time.
    
    Thus,~\Cref{thm:ER_SSS_pstar} can be converted to an LCA with the aforementioned runtime.
\end{proof}

We now prove Theorem~\ref{thm:ER_SSS_pstar}, which gives  a 4-round distributed algorithm
for access to a connected subgraph of the input graph $G$. We begin by defining quantities used in the algorithm:
\begin{definition}
    For a vertex $v$, let $\sm(v)$ and $\lr(v)$ be the smallest and largest index vertices connected to $v$ in $G$ (not including $v$ itself).
\end{definition}
Next, we recall types of vertices in the algorithm.
\begin{definition}
    If $v$ is such that $v<\sm(v)$, we call it a \textbf{candidate leader}, and otherwise call it a \textbf{non-candidate}.
    If $v$ is a candidate leader, and additionally its index is smaller than all of its $2$-hop neighbors, then we call $v$ a \textbf{leader}.
    If $v$ is a leader, we define its \textbf{administrator} $\admin(v) \leftarrow \lr(v)$ to be its largest-index neighbor.
    Let $\cL$ be the set of all leaders. 
\end{definition}
Once we have performed the first two rounds of the algorithm, we decompose the graph into a set of directed trees, where the root of each tree is the administrator chosen by the leader vertex.
\begin{definition}
	For $G\la \Gnps$, its \textbf{sparsified subgraph} $H$ is a random graph that is the deterministic output of~Algorithm~\ref{alg:ER_SSS_pstar} on input $G$.
\end{definition}

In order to lower bound the probability of $H$ being connected, we will define sets of events called \textbf{baseline conditions},
and argue that each base graph $G$ satisfies exactly one baseline condition. Then we prove the desired properties of the algorithm, conditioned on the baseline condition obeying certain properties (and we show that these properties are satisfied with high probability).

A baseline condition is a minimal collection of statements of the form $(u,v) \in E$ or $(u,v) \notin E$.
Essentially, these are the edges ``viewed'' by the algorithm in the first three rounds.
A baseline condition uniquely fixes the set of leaders and the representative leader of each vertex.
\begin{definition}
	A \textbf{baseline condition} is an event parameterized as $B(f, g)$
	where 
	\[f: V \ra \{\bot\}\cup V, \qquad  g: \{v | f(v) = \bot \} \ra 2^{V}\] 
	are two functions. 
	The function $f$ maps vertices to their smallest neighbor if they are not a candidate leader, and $\bot$ otherwise, and $g$ maps candidate leaders to the set of their neighbors. We say that a graph $G$ \textbf{satisfies} $B$ if it is compatible with $B$ in the obvious way. 
\end{definition}
Note that $f(v)=i$ implies that for every graph $G$ satisfying $B$, we have $(v,i)\in G$ and $(v,j)\notin G$ for every $j<i$. However, all edges from $v$ that are not fixed by values of $f,g$ on other vertices are present in a random graph $G$ satisfying $B$ independently with probability $p$.

We require several properties of these baseline conditions, which we will now prove. First, they partition the space of all possible graphs.
\begin{restatable}{lemma}{baseline_condition_incompatible}
	\label{lem:baseline_condition_incompatible}
	For every pair of distinct baseline conditions $B_1=B(f_1,g_1)$ and $B_{2}=B(f_2,g_2)$, there is no graph $G$ that satisfies both.
\end{restatable}
\begin{proof}
	There must be some vertex $v$ such that $f_1(v)\neq f_2(v)$ or $f_1(v)=f_2(v)$ and $g_1(v)\neq g_2(v)$, and fix such a vertex. In the first case, without loss of generality $f_1(v)=u\neq \perp$. In this case, every graph satisfying $B_1$ must have $v$'s smallest neighbor $u$, and hence it cannot satisfy $B_2$. In the second case, every graph satisfying $B_1$ must have $\Gamma(v)=g_1(v)\neq g_2(v)$, and hence it cannot satisfy $B_2$.
\end{proof}

Furthermore, all graphs $G$ satisfying $B$ have the same set of leaders and administrators.
\begin{lemma}\label{lem:base_unique}
	Every baseline condition $B$ uniquely specifies the set of candidate leaders, leaders, and administrators for every graph satisfying $B$. 
\end{lemma}
\begin{proof}
    One can see that the information provided by $B(f,g)$ suffices to run~Algorithm~\ref{alg:ER_SSS_pstar} for the first three rounds, and these rounds uniquely determine the set of administrators.
\end{proof}
Due to~\Cref{lem:base_unique}, we refer to the leaders and administrators of $B$ as the unique set of leaders and administrators of any graph satisfying $B$.

We now define favorable graphs and favorable baseline events. A favorable graph is one that induces a favorable event, and a favorable event is one such that a random graph satisfying it results in the algorithm succeeding with high probability (over the graph).

\begin{definition}
	\label{def:pis}
	Define $\Pi = \Pi_{1} \bigcap \Pi_{2}$ to be the event that $G\la \Gnps$ simultaneously satisfies all of the following conditions:
	\begin{itemize}
		\item $\Pi_{1}$: All $v \geq \frac{n}{3}$ are not candidate leaders;
		\item $\Pi_{2}$: Each $v < \frac{n}{3}$ has at least one neighbor whose index is greater than or equal to $\frac{2n}{3}$. 
	\end{itemize}
	If $G$ satisfies $\Pi$, it is called a \textbf{\emph{favorable graph}}.
\end{definition}

\begin{definition}
	A baseline condition $B(f, g)$ is called a \textbf{\emph{favorable condition}} if all of the following conditions hold:
	\begin{itemize}
		\item $\Phi_{a}$: All leaders of $B(f, g)$ have index less than $\frac{n}{3}$; 
		\item $\Phi_{b}$: For every $v$ such that $\frac{n}{3} \leq v < \frac{2n}{3}$, $v$ is not a candidate leader (i.e. $f(v) \neq \bot$);
		\item $\Phi_{c}$: All administrators of $B(f, g)$ have index at least $\frac{2n}{3}$.
	\end{itemize}
\end{definition}
As motivation for the definition, we desire the event $B(f,g)$ to leave the administrators with many unfixed edges, such that these administrators connect all subtrees with high probability. 

\begin{restatable}{lemma}{favorable_implication}
	\label{lem:favorable_implication}
	If an instance of $G$ is a favorable graph, then the baseline condition that $G$ satisfies is a favorable condition.
\end{restatable}
\begin{proof} 
	Notice that $\Pi_{1}$ immediately implies $\Phi_{a}$ and $\Phi_{b}$, and $\Pi_{2}$ immediately implies $\Phi_{c}$.
\end{proof}

Furthermore, a graph is favorable with high probability. 
\begin{restatable}{lemma}{FavGraph}
	\label{lem:favGraph}
	We have that $G\la \Gnps$ is a favorable graph with probability at least $1-n^{-4/3}$.
\end{restatable}
\begin{proof}
	If for every vertex $v$, $v$ retains an edge to $\{1,\ldots,n/3\}$ and $\{2n/3,\ldots,n\}$, we have a favorable graph. This occurs for an arbitrary fixed $v$ with probability
	\[
	2(1-\ps )^{n/3}\le 2e^{-\ps n/3}
	\]
	and hence the total probability of a failure is at most $n2e^{-\ps n/3}\le n^{-4/3}$.
\end{proof}
Finally, for every favorable baseline condition, a random graph satisfying this condition induces a connected graph with high probability.
\begin{restatable}{lemma}{FavCond}
	\label{lem:favCond}
	For every favorable condition $B(f, g)$, 
	\[
	\Pr_{G\la \Gnps}[H \text{ is connected }|G\text{ satisfies $B$}] \ge 1-2n^{-4/3}.
	\]
\end{restatable}
We prove~\Cref{lem:favCond} in the following subsection. Using the lemma, we can show that the connectivity guarantee is satisfied.

\begin{restatable}{lemma}{ERSSSconn}
	\label{lem:ER_SSS_conn}
	Let $H$ be the output of~Algorithm~\ref{alg:ER_SSS_pstar} on $G\la \Gnps$. We have that $H$ is connected with probability at least $1 - \frac{1}{n}$.
\end{restatable}
\begin{proof}
	We have that
	\begin{align*}
		\Pr_{G\la \Gnps}[H \text{ connected}] &= \sum_{B(f,g)}\Pr[H \text{ connected }|G\text{ satisfies $B$}]\Pr[H \text{ satisfies B}]\\
		&\ge \sum_{\text{fav. }B(f,g)}\Pr[H \text{ connected }|G\text{ satisfies $B$}]\Pr[H \text{ satisfies B}]\\
		&\ge (1-2n^{-7/6})\sum_{\text{fav. }B(f,g)}\Pr[H \text{ satisfies B}] && \text{(\Cref{lem:favCond})}\\
		&\ge (1-2n^{-7/6})(1-n^{-4/3}) && \text{(\Cref{lem:favGraph})}\\
        &\ge 1-1/n && \qedhere
	\end{align*}
\end{proof}

\subsection{Proof of~\ts{\Cref{lem:favCond}}{Connectivity Given Favorable}}
We show~\Cref{lem:favCond} by identifying a large set of possible edges that are each present with probability $\ps $ for a random $G$ satisfying an arbitrary favorable condition, and the presence of (a small number of) these edges establishes connectivity between all subtrees.

\newcommand{\med}{\cM}
For the remainder of the subsection fix an arbitrary favorable condition $B(f,g)$, and recall that $\cA$ is the set of leaders. Let $\med=\{n/3,\ldots,2n/3\}\subseteq V$ be the set of vertices with \textbf{medium} index. 
\begin{definition}
	Fixing $B(f,g)$, for every $a\in \cA$ let $T(a)$ be the set of vertices  connected to $a$ through the first three rounds of the algorithm (and note that this set is uniquely determined given $B$). 
\end{definition}

We now define the set of edges that are still unfixed in every favorable condition.

\begin{restatable}{lemma}{conn_independence}
	\label{lem:cav_indep}
	For every administrator $a\in \cA$ and medium vertex $v\in \med$, let $C_{a,v}$ be the event that $(a,v)$ is in $G$. Then $C_{a,v}$ occurs independently with probability $\ps $.
\end{restatable}
\begin{proof}
	Recall that $B(f,g)$ is a favorable condition, and hence:
	\begin{itemize}
		\item $\Phi_{a}$ holds and so all candidate leaders have index at most $n/3$,
		\item $\Phi_{b}$ holds so all $v \in \med$ are not candidate leaders (and hence do not have their neighborhoods specified by $g$),
		\item $\Phi_{c}$ implies $\min(\cA) \ge \frac{2n}{3}$ (and hence no administrator is a candidate leader).
	\end{itemize} 
	For arbitrary $a,v$ we have that $a\ge 2n/3 > n/3$ and $v>n/3$ and hence neither $a$ nor $v$ is a candidate leader (and hence $g$ does not determine the full neighborhood of either vertex). Moreover, both have some neighbor with index at most $n/3$, and hence the status of edge $(a,v)$ is not determined by the value of $f$. Thus, the edge is not conditioned on by $B$, and hence occurs independently with probability $\ps $.
\end{proof}

Now we show that the graph is connected with high probability. 
\newcommand{\Split}{\mathrm{Split}}
\begin{definition}
	For every partition $A_1,A_2$ of $\cA$ with $|A_1|\le |\cA|/2$, let $\Split(A_1)$ be the event that there are no edges in $G$ from $A_1$ to $T(A_2)\cap \med$  and from $A_2$ to $T(A_1)\cap \med$.
\end{definition}
Observe that $\cap_{A_1}\neg \Split(A_1)$ suffices for the graph to be connected:
\begin{lemma}\label{lem:split_nec}
	For every $G$ satisfying $B$ such that none of $\Split(A_1)$
    occurs, we have that $H$ is connected.
\end{lemma}
\begin{proof}
	Fix an arbitrary cut $V_1,V_2$. If the cut bisects any tree $T(a)$ we are clearly done by the edges added in the first three rounds, so WLOG assume this does not occur and let $A_1=V_1\cap \cA$ and $A_2=V_2\cap \cL$. As $\Split(A_1)$ does not occur, there is some edge from $a\in A_1\subseteq V_1$ to $T(A_2)\subseteq V_2$ or an edge from $a' \in A_2\subseteq V_1$ to $T(A_1)\subseteq A_1$, and such an edge is retained in round 4 of the algorithm, so we preserve connectivity across the cut.
\end{proof}

Finally, we show that no such event occurs with more than negligible probability.
\begin{restatable}{lemma}{split_rare}
	\label{lem:split_rare}
	For every $A_1$ we have
	\[\Pr[\Split(A_1)] \le  (1-\ps )^{\frac{n}{6}|A_1|}.
	\]
\end{restatable}
\begin{proof}
	Either $T(A_1)\cap \med \ge n/6$ or $T(A_2)\cap \med \ge n/6$. WLOG supposing the latter occurs (as otherwise the bound is only stronger), we have that all events $C_{a,v}$ for $a\in A_1, v \in T(A_2)\cap \med$ imply the negation of $\Split(A_1)$, and moreover there are at least $|A_1|\cdot (n/6)$ such events. By~\Cref{lem:cav_indep} each such event occurs independently with probability $\ps $, so the bound is as claimed.
\end{proof}

We then recall the lemma to be proved.
\FavCond*
\begin{proof}
	We have that
	\begin{align*}
		\Pr[G \text{ is disconnected}|G\text{ satisfies $B$}] &\le \sum_{A_1\subseteq \cA} \Pr[\Split(A_1)]&& \text{(\Cref{lem:split_nec})}\\
        &\le \sum_{A_1\subseteq \cA} (1-p)^{\frac{n}{6}|A_1|} && \text{(\Cref{lem:split_rare})}\\
        &\le \sum_{A_1\subseteq \cA}e^{-(\ps n/6)|A_1|} \le 2n^{-7/6}. && \qedhere
	\end{align*}
\end{proof}

\subsection{Bounding Subgraph Size}
\label{sec:proof_of_size}

Now we prove that the size of $H$ is as claimed. 
\begin{restatable}{lemma}{ERSSSsize}
	\label{lem:ER_SSS_size}
	Let $H$ be the output of~Algorithm~\ref{alg:ER_SSS_pstar} on $G\la \Gnps$. We have that $H$ has size $n+o(n)$ with probability $1 - \frac{1}{n}$.
\end{restatable}
First, note that every non-leader adds exactly one edge to $H$, and each leader $v$ with administrator $a$ adds exactly $|\Gamma(a)|$ edges. Then the final size bound is at most $\Gamma_{\max}\cdot |\cA|$. As the first is $O(\log n)$ whp, it suffices to show $|\cA|\le |\cL|=O(n/\log^2n)$. Note that a vertex is a leader if and only if is has the least rank among its $2$-hop neighborhood, which follows from simple concentration bounds.

\begin{fact}\label{fct:size}
	With high probability, $|\Gamma(v)|\le 14\log n$ and $|\Gamma(\Gamma(v))|\ge \log^2(n)/4$ for every $v$.
\end{fact}
Next, we bound the number of leaders
We also have the following bound on the number of index-leaders which arises from
a bound on the size of the 2-hop neighborhood of all vertices.

\begin{restatable}{lemma}{leaderBound}
	\label{lem:leader_bound}
	With high probability, there are at most $O(n/\log^2 n)$ leaders.
\end{restatable}
\begin{proof}
    By~\Cref{fct:size} we have that with high probability every vertex has at least $\Omega(\log^2 n)$ $2$-hop neighbors. Since a vertex being a leader implies no other member of its $2$-hop neighborhood can be a leader, (as that vertex would not have minimal index in its two-hop neighborhood) we are done.
\end{proof}
Putting the two results together implies the claimed size bound:
\begin{proof}[Proof of~\Cref{lem:ER_SSS_size}]
    It is clear that the first 3 rounds add at most $n-1$ edges, so we have that the edge count is bounded by $n$ plus
    \[
    \sum_{a\in \cA}|\Gamma(a)| \leq O(\log n)\cdot |\cA|\le O(\log n)\cdot |\cL| \le O(n/\log n)=o(n)
    \]
    where the first inequality is~\Cref{fct:size}, the second is that each administrator can be associated with at least one leader, and the third is~\Cref{lem:leader_bound}.
\end{proof}

\subsection{Distributed Erd\H{o}s-R\'{e}nyi Sparsifier}
\label{app:distrib_algo}

\begin{algorithm}[H]
\label{alg:ER_SSS_pstar}
\SetAlgoLined
 Initialize $H = \{ \}$ \\
 Round 1: 
 \For{$v \in V$} {
   %$v$ uniformly samples an integer $Index(v)$ from $\{0, \cdots, n^{4} -1\}$
   $v$ sends its index to all of $v$'s neighbors 
 }
 Round 2: 
    \For{$v \in V$} {
    \eIf{$v<u$ for all $u\in \Gamma(v)$}{
        $v$ nominates itself as a candidate leader

        }{
        $v$ marks itself as a non-candidate \\
        %$v$ sends ``you could be a leader'' \footnote{It is worth nothing that the messages in the algorithm can be represented by a single bit, but we wrote out the message for clarity.} to $Smallest(v)$ \\
        $v$ sends ``you are not a leader'' to all $\Gamma(v)\setminus \sm(v)$ \\
        $v$ adds edge $(v,\sm(v))$ to $H$ \\
        
        }
 }
 Round 3: 
    \For{$v \in V$, $v$ is a candidate leader} {
    \eIf{$v$ received at least one ``you are not a leader'' in Round 2}{
        $v$ adds edge $(v, u)$ to $H$, where $u$ is the smallest neighbor of $v$ that sent $v$ ``you are not a leader''

        }{
            $v$ elects itself as a leader \\
        $v$ sends ``you are my administrator'' to $\lr(v)$
        
        }
 }
 Round 4: 
    \For{$v \in V$, $v$ received ``you are my administrator'' in round 3} {
        $v$ sets itself as an administrator
        
        \For{$u \in \Gamma(v)$} {
        $v$ adds $(u, v)$ to $H$
        }
 } 

 \caption{4-round Distributed Sparsified Connected Subgraph Algorithm}
\end{algorithm}

%% file: 700-PA-spanners.tex
\section{Spanners on Preferential Attachment Graphs}\label{app:PA}

\subsection{The Model and Related Work}

We formally define the preferential attachment model:
\begin{definition}[Preferential Attachment graph]\label{def:PA}
For a function $\mu=\mu(n)$, we say $G\la \Gpa$ is a random \textbf{Preferential Attachment graph} if it is constructed as follows:
\begin{itemize}
    \item On round $1 \leq i \leq n$, add a vertex $v_{i}$ into the graph.
    \item Then, repeat the following process $\mu$ times:
    \begin{itemize}
        \item Add an edge from $v_{i}$ to a random vertex $v_j$
            (potentially a self loop to $v_{i}$ itself)
        \item The probability of the edge $(v_i, v_j)$ being added is
            $\frac{d_{v_{j}}}{\sum_{j=1}^{i}d_{v_{j}} + 1}$ if $j \neq i$,
            and $\frac{d_{v_{j}} + 1}{\sum_{j=1}^{i}d_{v_{j}} + 1}$ if $j = i$.
    \end{itemize}
\end{itemize}
\end{definition}
For clarity, in this section we let $v_i$ for $i\in [n]$ denote the vertex that was added in round $i$. Note that the initialization of the preferential attachment graph, by the definition above, is one vertex $v_1$ with $\mu$ self loops.

We now describe the connections to and differences with related work. Our result is incomparable to prior work~\cite{PA1,PA2}, which constructs an LCA\footnote{Their results are written as local information algorithms (LIAs), which are sublinear algorithms with the restriction that all queries are adjacent to already explored nodes. Our results for the preferential attachment model likewise obey this restriction. The LIAs constructed in prior work find a path from any node to the root node and therefore can be interpreted as LCAs.} for the \textbf{root-finding problem}, that of identifying a path to the first node. Their algorithm for that problem immediately implies the following:
\begin{theorem}[Implied by \cite{PA2}]
    Let $\nu=\omega(\log n)$ be any function that dominates $\log(n)$. 
    For every preferential attachment process with parameter $\mu$, there is an average-case LCA for $G\la \Gpa$ that w.h.p gives access to a $\nu$-spanner $H\subseteq G$. Moreover:
    \begin{itemize}
        \item on query $(u,v)$ the LCA has time complexity $O((d_u+d_v)\nu)$, which is $O(\mu \sqrt{n})$ in the worst-case and $O(\mu \polylog(n))$ in expectation (over all possible queries),
        \item $H$ contains at most $\nu \cdot n$ edges.\footnote{This bound may not be tight, but their result does not seem to give sparsity smaller than $O(n\log n)$ in any case.}
    \end{itemize} 
\end{theorem}
In particular, note that for every degree parameter $\mu \ge \polylog(n)$, the algorithm obtains a $\tO(\log n)$-spanner with $\tO(n)$ edges. As long as the edge parameter is sufficiently large, we obtain an improved sparsity and query time bound.
\begin{remark}
    Their result is not optimized for the spanner problem, and our algorithm is essentially a simplified version of their approach tailored to this task. In particular, we may assume that the parameter $\mu$ is sufficiently large, as otherwise the graph is already sparse, whereas they solve the root-finding problem even for highly sparse graphs.
\end{remark}

\subsection{Main Structural Lemma}
\label{sec:proof_PAMainLemma}

The proof of \Cref{thm:PA_Spanner} relies on the following structural result about preferential attachment graphs:
\begin{restatable}{lemma}{PAMainLemma}
	\label{lem:PAMainLemma}
	For $\mu(n)\ge c_\mu \log n$ for $c_\mu$ an absolute constant, the following holds. With high probability, every vertex $v$ that is not the highest-degree vertex is either a neighbor of the highest-degree vertex, or $v$ has a neighbor $u$ such that $d_u > 2\cdot d_v$.
\end{restatable}

To prove this lemma, we first show a few other structural properties.
We first define the degree of the preferential attachment graph after each round of sampling.
\begin{definition}
    For every $t\in [n]$ let $d_{v; t}$ be the degree of $v$ after $t$ rounds of applying the preferential process.
    In particular, $d_{v;n} = d_{v}$.
    For a vertex set $S\subseteq [n]$, let $d_{S;t}=\sum_{v\in S}d_{v;t}$ and $d_S=d_{S;n}$. For convenience, we let $d_{v,\tau}=d_{v,\lfloor \tau \rfloor}$ for non-integer $\tau\le n$.
\end{definition} 
We note a basic fact that we will repeatedly use:
\begin{fact}\label{fct:PA}
    For every $i$ we have $d_{v_i,i} \le 2\mu$, and for every $j\ge i$ we have $d_{v_i,j}\ge \mu$.
\end{fact}
This follows from the fact that each vertex has degree zero immediately before it is added, and finishes the subsequent round with between $\mu$ and $2\mu$ edges (as in the worst case all added edges are self loops).

Next, we recall a concentration bound for the degrees in preferential attachment graphs of~\cite{DKR18}.
For a vertex set $S\subseteq [n]$, the lemma bounds the sum of the degrees of nodes in $S$ at time $n$ in terms of their total degree at time $t$ multiplied by scaling factor $\sqrt{\frac{n}{t}}$, and a small constant error (close to 1).

\begin{restatable}[Lemma 3.8~\cite{DKR18}]{lemma}{PrefBoundLemma}
\label{lem:PrefBoundLemma}
    Assume $\mu \ge c_\mu \log n$ for a global constant $c_\mu$. Then there exists a constant $c_{t} = 40^{6} + 1$ such that for every $t \in [c_{t},n]$ and $S \subseteq \{ v_{1}, \cdots, v_{t}\}$, we have:
    \[
    \Pr \left[ \frac{39}{40}\sqrt{\frac{n}{t}}d_{S;t} < d_{S} < \frac{41}{40}\sqrt{\frac{n}{t}}d_{S;t}  \right] \geq 1 - \frac{1}{n^{10}}.
    \]
\end{restatable}
We remark that this statement follows from their Lemma 3.8 with $\eps=1/40$ and $\mu\ge c_{\mu}\log n$, where $c_{\mu}$ is a large enough constant such that the failure probability becomes as claimed.

We first show that for large enough $i$, $v_i$ is directly connected to \textit{some} vertex $v_j$ where $j\ll i$. Note that $v_i$ is not connected to all vertices with substantially smaller index, but there is at least one neighbor with this property.
\begin{restatable}{lemma}{edgeEarlier}
\label{lem:edgeEarlier}
    Let $w=1/16$. With high probability, for every $i \ge c_{t}/w$ there is $j \le iw$ such that $(v_i,v_j)\in G$.
\end{restatable}
\begin{proof}
Let $S = \{v_{1}, \cdots, v_{w\cdot i}\}$. Then we have
\begin{align*}
    d_{S, i} &\geq  \frac{40}{41}\cdot \sqrt{\frac{i}{n}} d_{S, n}
    & \textrm{(RHS of~\Cref{lem:PrefBoundLemma} with $t=i$)}\\
    &\geq \frac{39}{41}\cdot \sqrt{\frac{i}{wi}} d_{S, wi}
    & \textrm{(LHS of~\Cref{lem:PrefBoundLemma} with $t=wi$)}\\
    &\ge 2\cdot d_{S,wi}\\
    &\ge 2\cdot \mu iw
\end{align*}
where the final step follows from $d_{S,wi}\ge |S|\mu$ by~\Cref{fct:PA} and that vertices in $S$ are added in rounds below $wi$.
Finally, this implies that for every edge added from $v_i$ in round $i$ goes to $S$ with probability at least $1-d_{S,i}/2\mu i = 1-w$, and hence the probability that none of the edges are adjacent to $S$ is at most $(1-w)^\mu \le n^{-100}$.
\end{proof}

Next, we show that for two vertices where one has substantially smaller index than the other, the smaller-index vertex has at least twice the degree.
\begin{restatable}{lemma}{earlierHeavier}
\label{lem:earlierHeavier}
    Again let $w=1/16$. With high probability, for every $i > c_{t}/w$ and $j < i\cdot w$ we have $d_{v_{i}} \le  d_{v_{j}}/2$.
\end{restatable}
\begin{proof}
For convenience, let $b=\max\{c_t,j\}$ and note that $d_{v_{i}, i} \leq 2\mu$ and $d_{v_j,\max\{c_t,j\}}\geq d_{v_{j},j} \geq  \mu$ by~\Cref{fct:PA}. 
We have
\begin{align*}
    d_{v_{i}} &\leq \frac{41}{40}d_{v_{i}, i} \sqrt{\frac{n}{i}} && \text{(RHS of~\Cref{lem:PrefBoundLemma} with $t=i$)}\\
    &\le \frac{41}{40}2 \mu \sqrt{\frac{n}{i}}\\
    &< \frac{39}{40} \frac{\mu}{2} \sqrt{\frac{n}{b}} && \text{($i>b/w$)}\\
    &\le \frac{39}{40}\frac{d_{v_{j}, b}}{2}\sqrt{\frac{n}{b}}\\
    &\le d_{v_{j}}/2  && \text{(LHS of~\Cref{lem:PrefBoundLemma} with $t=b$)} \qedhere
\end{align*}
\end{proof}
As an easy corollary, we obtain that the highest-degree vertex is in the first $c_t/w$ indices.
\begin{corollary}\label{cor:highestDeg}
     Again let $w=1/16$. With high probability, the highest degree vertex has index bounded by $c_t/w^2$.
\end{corollary}
\begin{proof}
    Fix $i\le c_t/w$ and $j\ge c_t/w^2$ arbitrarily. We claim that $d_{v_i}\ge d_{v_j}$ with high probability, which suffices to show the result. This is immediate from~\Cref{lem:earlierHeavier} (switching the roles of $j$ and $i$).
\end{proof}

Finally, we show that all small-index vertices are connected with high probability.
\begin{lemma}\label{lem:smallClique}
    Again let $w=1/16$. With high probability, for every $i,j \le c_t/w^2$ where $i\neq j$, we have $(v_i,v_j)\in G$.
\end{lemma}
\begin{proof}
    WLOG assume $i\le j$ and note that $d_{v_i,j}\ge \mu$. Then it is easy to see that for every edge inserted from $v_j$, it connects to $v_i$ with probability at least $\mu/(j+1)\mu\ge w^2/2c_t$. Therefore, the edge is not present with probability at least $(1-\frac{w^2}{2c_t})^\mu \le n^{-100}$, using that $\mu \ge c_\mu \log(n)$ is sufficiently large.
\end{proof}

We now prove the lemma.
\begin{proof}[Proof of~\Cref{lem:PAMainLemma}]
    Fix $v_i$ where $i$ is arbitrary. If $i\le c_t/w^2$, we have that $(v_i,v_j)$ are present in $G$ for every $j\le c_t/w^2$ by~\Cref{lem:smallClique}, and hence $v_i$ has an edge to the highest degree vertex by~\Cref{cor:highestDeg}. Otherwise we have $i>c_t^2$, and then by~\Cref{lem:edgeEarlier} there is an edge $(v_i,v_j)$ with $j<i\cdot w$, and moreover $d_{v_j}\ge 2\cdot d_{v_i}$ by~\Cref{lem:earlierHeavier}.
\end{proof}

\subsection{Proof of~\ts{\Cref{thm:PA_Spanner}}{Main Result}}
We now use~\Cref{lem:PAMainLemma} to prove~\Cref{thm:PA_Spanner}. The average-case LCA is the algorithm given in \Cref{sec:PA-overview}.

First, we establish a bound on the average query time of the algorithm, by bounding the degree of the small and large vertices.
\begin{claim}\label{clm:largeVertDeg}
    Let $S=\{v_1,\ldots,v_{c_t/w}\}$. With high probability, $ \sum_{v \in S}d_{v} = O(\sqrt{n}\cdot \mu)$.
\end{claim}
\begin{proof}
    Applying~\Cref{lem:PrefBoundLemma} with $t=c_t/w$ and $S=S$ and using that $d_{S;c_t/w}\le 2\mu c_t/w$ immediately gives the bound.
\end{proof}

\PASpanner*
\begin{proof}
    On query $(u,v)$, the algorithm checks if $v$ is the highest-degree neighbor of $u$ or vice versa, and if so keeps the edge and otherwise discards it. We first argue that this rule produces a connected subgraph of size $n-1$. By~\Cref{lem:PAMainLemma}, every vertex except that of highest degree keeps an edge. Moreover, from each vertex, let $h(v)$ be the highest degree neighbor. We have that either $h(v)=v_{\max}$, the highest degree vertex, or $d_{h(v)}\ge 2\cdot d_v$. Thus, for every vertex $v$ the path $(v,h(v),h(h(v)),\ldots)$ has length at most $\log(n\cdot \mu)$, and terminates at $v_{\max}$, with high probability. Thus, the constructed graph has diameter $O(\log n)$ as claimed.
    
    Finally, we argue that the average query time is as claimed. We have that the average query time is 
    \begin{align*}
        \frac{1}{|E|}\sum_{u,v}(d_u+d_v)^2 &\le \frac{2}{n\mu}\sum_v d_v^2\\
        &= \frac{2}{n\mu}\left(\sum_{i\le c_t/w}d_{v_i}^2 + \sum_{i> c_t/w}d_{v_i}^2\right)\\
        &\le \frac{2}{n\mu}\left( O(\sqrt{n}\cdot \mu)^2+\sum_{i> c_t/w}\left(c'2\mu \sqrt{\frac{n}{i}}\right)^2\right) && \text{(\Cref{clm:largeVertDeg,lem:PrefBoundLemma})}\\
        &= O(\mu \cdot \log^3 n). &&\qedhere
    \end{align*}
\end{proof}

%% file: 800-UA-spanners.tex
\section{Spanners on Uniform Attachment Graphs}\label{sec:UA_spanners}

\subsection{The Model}\label{subsec:UA_model}

Next, we construct low-stretch spanners for uniform attachment graphs with sufficiently high degree parameter $\mu \geq c_{\mu} \log(n)^2$, for some global constant $c_{\mu}$. In the generation of a uniform attachment graph, at each time step a node joins the graph and connects to $\mu$ existing nodes independently and uniformly at random. (See \Cref{def:UA}.) This contrasts preferential attachment, in which new nodes connect to existing nodes with probability proportional to their degrees. We prove that the same algorithm (with slight modification) used in the preferential attachment case can be applied in this setting. Similar guarantees for spanning and sparsity can be achieved, with improved guarantees regarding the amount of local work. 

We now formally define the uniform attachment model:
\begin{definition}[Uniform Attachment graph]\label{def:UA}
For a function $\mu=\mu(n)$, we say $G\la \Gua$ is a random \textbf{Uniform Attachment graph} if it is constructed as follows. 
\begin{itemize}
    \item On round $1$, add a vertex $v_1$ into the graph.
    \item On round $2 \leq i \leq n$, add a vertex $v_{i}$ into the graph. Then, choose $\mu$ vertices $\{v_j\}_{j \in [\mu]}$ independently and uniformly at random out of the existing vertices. Add an edge from $v_i$ to each vertex chosen.
\end{itemize}
\end{definition}

Note that the definition above does not allow for self loops. However, the properties used for uniform attachment graphs would still apply for the setting where self loops are allowed, and our proofs would extend with slight modifications to this setting.

The LCA for spanners for uniform attachment graphs is the following. Uniform attachment graphs are multigraphs, and so the algorithm's input specifies the edges adjacent to the edge as well as the edge's lexicographic indexing. (See \Cref{remark:multigraph}.)
\\\\
\noindent\textbf{Algorithm:} On input $(u, v, i)$: \\ 
Check if both $u$ and $v$ have degree greater than $\mu \cdot (H_{n-1} - H_6) + \mu/2$, where $H_n$ denotes the $n$-th Harmonic number. If this is the case, keep the edge $(u, v, i)$ if $i = 1$ (i.e. it is the lexicographically first edge between $u$ and $v$). Otherwise, check if $v$ is the highest-degree neighbor of $u$ or vice versa. If so, keep the edge if $i = 1$ and otherwise discard the edge.
\\\\
We now describe the necessity of considering $\mu \geq c_{\mu} \log(n)^2$ in our approach. Our algorithm uses that each node $v_i$ in the graph is connected to at least one node $v_j$ added sufficiently earlier in the process. Additionally, this node $v_j$ can be distinguished from the neighbors of $v_i$ added later in the process because $v_j$ will have a high degree and the neighbors added later will have a low degree. The connectivity property requires $\mu \geq c'_{\mu} \log(n)$ for some constant $c'_{\mu}$, and the property that early nodes can be distinguished with their degrees requires that $\mu \geq c_{\mu} \log(n)^2$ for some constant $c_{\mu}$. The case we consider of $\mu \geq c_{\mu} \log(n)^2$ can be contrasted with the case where $\mu$ is a small constant independent of $n$, for which it is known that the distribution of the degrees of individual nodes in the graph will not be well-concentrated enough to appropriately distinguish nodes added early from nodes added later on (see \cite{lodewijks2020maximal}, for example).  
For $\mu \geq c_{\mu} \log(n)^2$, we show that the degrees are concentrated enough to reveal information about the arrival times of nodes.

\subsection{Spanners for Uniform Attachment Graphs Given Times of Arrival}\label{subsec:UA-given-arrivaltimes}

We remark that if the LCA knew the arrival times of nodes and additional labeling is given to edges, an even simpler algorithm suffices for providing local access to an $O(\log n)$ spanner with $n - 1$ edges. This algorithm works for uniform attachment graphs with any parameter $\mu$.

Suppose that each node in the graph is labeled by its time of arrival. Suppose also that, when a new node $v$ joins the graph and forms $\mu$ edges, these edges are labeled with $v$ as well as numbers from $1$ to $\mu$. Any arbitrary ordering of the edge labels with numbers from $1$ to $\mu$ for the edges corresponding to $v$ suffices.

In this setting, consider the following spanner LCA: on input $(u, v)$, keep the edge if $u$ has an earlier arrival time than $v$ and the edge has label $\{\textsc{vertex}=v, \textsc{edge-number}=1\}$, or if $v$ has an earlier arrival time than $u$ and the edge has label $\{\textsc{vertex}=u, \textsc{edge-number}=1\}$. 

For each node $u$ in the graph, this algorithm keeps a uniform random edge out of the edges to nodes added earlier in the process. Consequently, the spanner keeps a random path from each node $u$ to the first node added. Let $R_t$ for $t \in [n]$ be the length of the shortest path from the node added at the $t$-th time step to the root node in the uniform attachment graph. As proven in \cite{devroye2011long}, for uniform attachment graphs with any parameter $\mu$,
    $$\mathbb{P}\left( \max_{1 \leq t \leq n} R_t > 2e \log(n)\right) \leq \frac{1}{n^3}.$$

Therefore, with high probability, this algorithm produces an $O(\log n)$ spanner with $n - 1$ edges. The amount of local work is $O(1)$, as the LCA only needs information about the labels of the vertices adjacent to the edge and the label of the edge.

\subsection{Main Structural Lemmas}\label{sec:ua-structural}

We now present the lemmas that we will need to prove \Cref{thm:UA_Spanner}. As mentioned in the proof overview in \Cref{sec:PA-overview}, we want to prove that every vertex is connected to at least one other vertex that is added sufficiently earlier in the uniform attachment process. To do so, we define \textit{intermediate centers} (as in the definition below) such that $C_1 \supset C_2 \supset \dots \supset C_M$ for some $M = O(\log n)$. We prove that vertices in $[n] \setminus C_1$ are connected to a vertex in $C_1$, and vertices with arrival time $i$ satisfying $|C_{m+1}| < i < e^2 \cdot |C_m|$ are connected to a vertex in $C_{m + 1}$, for $m \in \{1, 2, \dots, M\}$. The reason for the $e^2$ is that degrees are not precise indicators of times of arrival; therefore, we can identify a neighbor whose time of arrival is $\leq e^2 \cdot |C_{m+1}|$ instead of $\leq |C_{m+1}|$, and we need to account for this in the analysis.

\begin{definition}[Intermediate Centers]\label{def:intermediate-centers}
Define 
\begin{equation}\label{def:M}
    M = \frac{\ln\left(\frac{n}{e^2}\right)}{\ln\left(\frac{\mu}{3 \ln(n) \cdot e^2}\right)} ~~ \text{ and } ~~ |C_m| = \frac{n}{e^2} \cdot \left( \frac{3 \ln(n) \cdot e^2}{\mu}\right)^m
\end{equation}
for $m \in \left\{1, 2, \dots, M \right\}$. Note that, by definition of $M$, $|C_M| = 1$. Define the $m$-th intermediate center $C_m$ to be the first $|C_m|$ nodes added to the uniform attachment graph.
\end{definition}
We prove the following connectivity properties related to the intermediate centers.
\begin{lemma}\label{lemma:UA-connected-to-centers}
With probability at least $1 - \frac{2}{n^2}$, the following guarantees hold. First, all $i > |C_1|$ have an edge to some $j \leq |C_1|$. Second, for all $m \in \left\{1, 2, \dots, M -1 \right\}$, all $|C_{m + 1}| < i < e^2 \cdot |C_m|$ have an edge to some $j \leq |C_{m + 1}|$. Third, the first $7$ nodes added all have an edge to the root node.
\end{lemma}

Once we have established that each vertex is connected to at least one vertex that arrived sufficiently earlier, we argue that we can use degrees of vertices to locally identify which neighbors of a vertex arrived sufficiently earlier and which ones did not.

\begin{lemma} \label{lemma:UA-degree-concentration-bounds}
    Let $\mu := \mu(n) \geq c_{\mu} \log(n)^2$ for some global constant $c_{\mu}$. Consider any time $|C| \in \mathbb{N}$. Let $\lambda_{|C|} = \mu \cdot \left( H_{n-1} - H_{|C|-1}\right)$, where $H_n$ denotes the $n$-th Harmonic number. Then, if $|C| \geq 2$, with probability at least $1 - \frac{4}{n^2}$, all nodes $i \leq |C|$ have degree greater than $\lambda_{|C|} + \frac{\mu}{2}$ and all nodes $i \geq e^2 \cdot |C|$ have degree less than $\lambda_{|C|} + \frac{\mu}{2}$. Additionally, if $|C| = 1$, with probability at least $1 - \frac{4}{n^2}$, node $1$ has degree greater than $\lambda_{2} + \frac{\mu}{2}$ and all nodes $i \geq 7$ have degree less than $\lambda_{2} + \frac{\mu}{2}$.
\end{lemma}

\subsection{Proof of Lemma \ref{lemma:UA-connected-to-centers}}

We now prove Lemma \ref{lemma:UA-connected-to-centers}, which states that with high probability, each node is connected to a smaller intermediate center (which consists of nodes added before some time in the process).

\begin{proof}[Proof of Lemma \ref{lemma:UA-connected-to-centers}]

Let's prove the first part of the lemma. Consider any fixed $i > |C_1|$. Then
    $$\mathbb{P}\left( \text{no edge from $i$ to some $j \leq |C_1|$}\right) \leq \left(1 - \frac{|C_1|}{i - 1} \right)^{\mu} \leq \left(1 - \frac{|C_1|}{n} \right)^{\mu} \leq \exp\left(-\frac{|C_1| \cdot \mu}{n}\right).$$
    By definition of $|C_1|$, this probability is at most $\frac{1}{n^3}$.

We next prove the second part of the lemma. Consider any $|C_{m+1}| < i < e^2 \cdot |C_m|$. Then
    $$\mathbb{P}\left( \text{no edge from $i$ to some $j \leq |C_{m+1}|$}\right) \leq \left(1 - \frac{|C_{m+1}|}{i - 1} \right)^{\mu} \leq \left(1 - \frac{|C_{m+1}|}{e^2 \cdot |C_m|} \right)^{\mu} \leq \exp\left(-\frac{|C_{m+1}| \cdot \mu}{e^2 \cdot |C_m|}\right).$$
    By definition of $|C_{m+1}|$ and $|C_m|$, this is at most $\frac{1}{n^3}$.

    Let us now prove the third part of the lemma. Consider any $i \leq 7$. The probability that any such $i$ has no edge to the root is at most $(1 - 1/7)^{\mu} \leq \exp(-\mu/7) \leq 1/n^3$ given the lower-bound on $\mu$ of $c_{\mu} \cdot \log(n)^2$.

    By a union bound, the probability that there exists a node not satisfying (1), (2), or (3) is at most $\frac{2}{n^2}$. The $2$ comes from the fact that any $i \in \{|C_j| , |C_j| + 1, \dots e^2 \cdot |C_j| - 1\}$ for any $j \in \{1, 2, \dots, M\}$ appears in two of the three conditions, and any other $i \in [n]$ appears in one of the three conditions. Therefore, the desired guarantees hold with probability at least $1 - \frac{2}{n^2}$. 
\end{proof}

\subsection{Proof of Lemma \ref{lemma:UA-degree-concentration-bounds}}

We now prove Lemma \ref{lemma:UA-degree-concentration-bounds}, which states that the degrees of individual nodes are sufficiently regular and concentrated when $\mu > c_{\mu} \log(n)^2$. We consider nodes that arrive before some time $|C|$ or after time $e^2 \cdot |C|$, where $C$ stands for ``center'' and $|C|$ is the size of the center. We prove that there is a corresponding degree threshold $\lambda_{|C|}$ such that all nodes added before time $|C|$ have degree $> \lambda_{|C|} + \mu/2$ and all nodes added after time $|C|$ have degree $< \lambda_{|C|} + \mu/2$, with high probability.

We will prove this lemma by relating the out-degrees of nodes to Poisson random variables with different parameters. We will then prove the concentration of the degrees by utilizing the concentration of Poisson random variables, specifically the following result. Let $\text{Poi}(\lambda)$ be a Poisson random variable with parameter $\lambda$.

\begin{lemma}[Poisson concentration \cite{canonnepoissonnote}]\label{lemma:poisson-concentration}
Let $X \sim \text{Poi}(\lambda)$ where $\lambda > 0$. For any $d > 0$,
$$\mathbb{P}\left( X \geq \lambda + d \right) \leq e^{\frac{-d^2}{2 (\lambda + d)}} ~~ \text{ and } ~~ \mathbb{P}\left( X \leq \lambda - d \right) \leq e^{\frac{-d^2}{2 (\lambda + d)}}.$$
\end{lemma}

We now prove Lemma \ref{lemma:UA-degree-concentration-bounds}.

\begin{proof}[Proof of Lemma \ref{lemma:UA-degree-concentration-bounds}]
    We break this proof into three steps. First, we connect the degrees of nodes to Poisson random variables. Second, we bound the tails of these random variables. Third, we tie this all together to prove the lemma.

    \paragraph{Step 1: Connecting the degree to a Poisson distribution.}
We express the degree distribution of the $i$th node in terms of random variables. It has previously been observed \cite{bollobas2001degree, mahmoud2014degree} that the degree distribution of each node converges to a certain Poisson distribution as $n \to \infty$, but we need to refine this further to make statements about the degrees at fixed $n$.

At each time step $j \geq i + 1$, the $s$-th edge (for $s \in [\mu]$) added from the new node connects to $i$ with probability $\frac{1}{j-1}$. Each of the $\mu$ connections that the new node makes is chosen independently. Each node $i \geq 2$ has $\mu$ edges that it added when it arrived, plus any edges that connected to it due to later nodes. Let the edges due to later nodes be called the outdegree of $i$.

Therefore, the outdegree $\text{outdeg}(i)$ of the $i$th node for $i \geq 2$ is distributed according to the following sum of random variables:
\begin{equation}\label{outdegree-i-geq-2}
    \text{outdeg}(i) \sim \sum_{j = i + 1}^n \sum_{s = 1}^{\mu} \text{Bern}\left(\frac{1}{j-1} \right),
\end{equation}
where $\text{Bern}\left(\frac{1}{j-1} \right)$ is a Bernoulli random variable with success probability $1/(j-1)$.

Note that, because the first node doesn't arrive with any edges, its degree $\text{deg}(1)$ equals its outdegree $\text{outdeg}(1)$ and is distributed as:
\begin{equation}\label{outdeg-root}
    \text{deg}(1) = \text{outdeg}(1) \sim \sum_{j = 2}^n \sum_{s = 1}^{\mu} \text{Bern}\left(\frac{1}{j-1} \right) = \mu + \sum_{j = 3}^n \sum_{s = 1}^{\mu} \text{Bern}\left(\frac{1}{j-1} \right).
\end{equation}

The outdegree of the $i$th node is a sum of independent Bernoulli random variables and is therefore distributed according to a Poisson binomial distribution, by definition. Let $\lambda_i = \sum_{j = i + 1}^n \sum_{s = 1}^{\mu} \frac{1}{j - 1}$. Note that $\lambda_i = \mu \cdot \left( H_{n-1} - H_{i-1}\right)$, where $H_n$ denotes the $n$-th Harmonic number. 

We use the following fact from Borisov and Ruzankin \cite[Lemma 2]{borisov2002poisson}. Recall that $\text{outdeg}(i)$ follows a Poisson binomial distribution and $\max \left\{ \frac{1}{i}, \frac{1}{i + 1}, \dots, \frac{1}{n-1}\right\}  = \frac{1}{i}$. Then for $i \geq 2$ and any $d$:
$$\mathbb{P}\left(\text{deg}(i) \leq \lambda_i - d + \mu\right)  = \mathbb{P}\left(\text{outdeg}(i) \leq \lambda_i - d \right) \leq \frac{\mathbb{P}\left(\text{Poi}(\lambda_i) \leq \lambda_i - d \right)}{(1 - \frac{1}{i})^2}  \leq  4 \cdot \mathbb{P}\left(\text{Poi}(\lambda_i) \leq \lambda_i - d \right).$$
Similarly, 
$\mathbb{P}\left(\text{deg}(i) \geq \lambda_i + d + \mu\right) = \mathbb{P}\left(\text{outdeg}(i) \geq \lambda_i + d \right) \leq  4 \cdot \mathbb{P}\left(\text{Poi}(\lambda_i) \geq \lambda_i + d \right).$

Using Equation (\ref{outdeg-root}), note that, similarly:
$$\mathbb{P}\left(\text{deg}(1) \leq \lambda_2 - d + \mu \right) = \mathbb{P}\left(\text{outdeg}(1) \leq \lambda_2 - d + \mu \right) \leq  4 \cdot \mathbb{P}\left(\text{Poi}(\lambda_2) \leq \lambda_2 - d \right)$$ 
and $$\mathbb{P}\left(\text{deg}(1) \geq \lambda_2 + d + \mu \right) = \mathbb{P}\left(\text{outdeg}(1) \geq \lambda_2 + d + \mu \right) \leq  4 \cdot \mathbb{P}\left(\text{Poi}(\lambda_2) \geq \lambda_2 + d \right).$$

\paragraph{Step 2: Tail bounds of the Poisson distribution.} Next, we prove both parts of the Lemma by bounding the left and right tail probabilities of the Poisson distributions. First consider $i \leq |C|$. Note that for all such $i$, $\lambda_i \geq \lambda_{|C|}$. By \Cref{lemma:poisson-concentration}, for $i \geq 2$,
\begin{equation}\label{right-tail-bound-poisson}
    \mathbb{P}\left(\text{Poi}(\lambda_i) \leq \lambda_{|C|} - \frac{\mu}{2} \right) = \mathbb{P}\left(\text{Poi}(\lambda_i) \leq \lambda_i - \left(\lambda_i - \lambda_{|C|} + \frac{\mu}{2} \right)\right)
\leq \exp\left(\frac{-\left(\lambda_i - \lambda_{|C|} + \frac{\mu}{2}\right)^2}{2(2 \lambda_i - \lambda_{|C|} + \frac{\mu}{2})} \right)
\end{equation}
$$\leq \exp\left(\frac{-\left(\frac{\mu}{2}\right)^2}{2(\lambda_2 + \frac{\mu}{2})} \right) = \exp\left(\frac{-\left(\frac{\mu}{2}\right)^2}{2(\mu \cdot (H_{n-1} -1 ) + \frac{\mu}{2})} \right) = \exp\left(\frac{-\mu}{8 \cdot (H_{n-1} - \frac{1}{2})} \right) \leq \frac{1}{n^3},$$
where the last expression uses that $\mu \geq c_{\mu} \log(n)^2$.

Next consider $i \geq e^2 \cdot |C|$. Note that for all such $i$, $\lambda_i \leq \lambda_{e^2 \cdot |C|}$. By \Cref{lemma:poisson-concentration},
\begin{equation}\label{late-node-degree-tailbound}
\mathbb{P}\left(\text{Poi}(\lambda_i) \geq \lambda_{|C|} - \frac{\mu}{2} \right) =\mathbb{P}\left(\text{Poi}(\lambda_i) \geq \lambda_i + \left(\lambda_{|C|} - \lambda_i - \frac{\mu}{2}\right) \right)
    \leq \exp \left(\frac{-\left(\lambda_{|C|} - \lambda_i - \frac{\mu}{2}\right)^2}{2 \left(\lambda_i + \left(\lambda_{|C|} - \lambda_i - \frac{\mu}{2}\right) \right)} \right).
\end{equation}
Note that $\lambda_{|C|} - \lambda_i \geq \lambda_{|C|} - \lambda_{e^2 \cdot |C|} \geq \mu$. Therefore, $\left(\lambda_{|C|} - \lambda_i - \frac{\mu}{2}\right) \geq \frac{\mu}{2}$, and Equation (\ref{late-node-degree-tailbound}) is bounded above by:
$$\leq \exp\left( \frac{-\left(\frac{\mu}{2}\right)^2}{2 \lambda_{|C|}}\right) = \exp\left( \frac{-\mu}{8 \left(H_{n-1} - H_{|C| - 1} \right)}\right) \leq \frac{1}{n^3},$$
where the last expression uses that $\mu \geq c_{\mu} \log(n)^2$.

\paragraph{Step 3: Putting everything together.}
Let us first prove the result for $|C| \geq 2$. We want to say that all nodes $i \leq |C|$ are sufficiently low-degree with high probability and all nodes $i \geq e^2 \cdot |C|$ are sufficiently high-degree with high probability.

Let's first focus on $i \leq |C|$. For $i = 1$, combining the steps above we see:
\begin{equation}\label{C-geq-2-bound}
    \mathbb{P}\left(\text{deg}(1) \leq \lambda_{|C|} + \frac{\mu}{2} \right) \leq 4 \cdot \mathbb{P}\left(\text{Poi}(\lambda_2) \leq \lambda_{|C|} - \frac{\mu}{2} \right) \leq \frac{4}{n^3}.
\end{equation}
For $2 \leq i \leq |C|$, we see:
$$\mathbb{P}\left(\text{deg}(i) \leq \lambda_{|C|} + \frac{\mu}{2} \right) \leq 4 \cdot \mathbb{P}\left(\text{Poi}(\lambda_i) \leq \lambda_{|C|} - \frac{\mu}{2} \right) \leq \frac{4}{n^3}.$$ 
Next, let's look at $i \geq e^2 \cdot |C|$. From the steps above, we have:
$$\mathbb{P}\left(\text{deg}(i) \geq \lambda_{|C|} + \frac{\mu}{2} \right) \leq 4 \cdot \mathbb{P}\left(\text{Poi}(\lambda_i) \geq \lambda_{|C|} - \frac{\mu}{2} \right) \leq \frac{4}{n^3}.$$ 
Therefore, we have achieved the desired concentration bounds on the degrees of nodes. By taking a union bound, we find that with probability at least $1 - \frac{4}{n^2}$, all $i \leq |C|$ have degree $> \lambda_{|C|} + \frac{\mu}{2}$ and all $i \geq e^2 \cdot |C|$ have degree $< \lambda_{|C|} + \frac{\mu}{2}$.

Let us now prove the result for $|C| = 1$. The reason we need to handle this case separately is because Equation (\ref{right-tail-bound-poisson}) required that $\lambda_i \geq \lambda_{|C|}$ to hold. If $|C| = 1$, but the concentration bounds for the degree of node $1$ are computed using a $\text{Poi}(\lambda_2)$ distribution, $\lambda_2 \geq \lambda_{|C|}$ no longer holds. Instead, by the computations above, we find that:
$$\mathbb{P}\left(\text{deg}(1) \leq \lambda_{2} + \frac{\mu}{2} \right) \leq 4 \cdot \mathbb{P}\left(\text{Poi}(\lambda_2) \leq \lambda_{2} - \frac{\mu}{2} \right) \leq \frac{4}{n^3}.$$
As before, for $i \geq e^2$ we have $\mathbb{P}\left(\text{deg}(i) \geq \lambda_{2} + \frac{\mu}{2} \right) \leq 4 \cdot \mathbb{P}\left(\text{Poi}(\lambda_i) \geq \lambda_{2} - \frac{\mu}{2} \right) \leq \frac{4}{n^3}.$ By taking a union bound, we find that with probability at least $1 - \frac{4}{n^2}$, the first node has degree $> \lambda_{2} + \frac{\mu}{2}$ and all $i \geq e^2$ have degree $< \lambda_{2} + \frac{\mu}{2}$.
\end{proof}

\subsection{Proof of Theorem \ref{thm:UA_Spanner}}

We are now ready to prove Theorem \ref{thm:UA_Spanner}.

\begin{proof}[Proof of Theorem \ref{thm:UA_Spanner}]
    \textit{Proof of $O(\log n)$ spanning.}
        Consider the event that the degree bounds from Lemma \ref{lemma:UA-degree-concentration-bounds} hold for all intermediate centers $C_m$ (Definition \ref{def:intermediate-centers}) and also the connectivity bounds from Lemma \ref{lemma:UA-connected-to-centers} hold. This event takes place with probability at least $1 - \frac{1}{n}$.

        We argue that, if this event holds, any two nodes $u$ and $v$ in the graph will have a path of length at most $O(\log n)$ between them in the subgraph that the LCA gives access to, which implies that the subgraph is an $O(\log n)$ spanner.

        Consider any node $u$ in the uniform attachment graph. We prove that the algorithm keeps a path of length $O(\log n)$ to the root (the first node added) of the uniform attachment graph. Let $u_0 := u$, and let $u_t$ denote the node in the path kept from $u$ to the root which is $t$ edges away from $u$. It must be the case that $u_0 > |C_1|$ or $|C_{m+1}| < u_0 \leq |C_{m}|$ for some $m \in \{1, 2, \dots, M-1\}$ and $M$ and $|C_m|$ as defined in Equation (\ref{def:M}). Suppose that $u_0 > |C_{m+1}|$ and that $m$ is the highest value such that this expression is satisfied. By Lemma \ref{lemma:UA-connected-to-centers}, $u_0$ must have an edge to some $j \leq |C_{m+1}|$. Consider the definition of $\lambda_{|C|}$ from Lemma \ref{lemma:UA-degree-concentration-bounds}. By Lemma \ref{lemma:UA-degree-concentration-bounds}, all neighbors of $u_0$ that arrived after time $e^2 \cdot |C_{m+1}|$ must have degree $< \lambda_{|C_{m+1}|} + \frac{\mu}{2}$ and all neighbors of $u_0$ that arrived before time $|C_{m+1}|$ must have degree $> \lambda_{|C_{m+1}|} + \frac{\mu}{2}$. Put together this implies that the highest-degree neighbor of $u_0$ must have arrived before time $e^2 \cdot |C_{m+1}|$. Therefore, because the LCA keeps the edge between $u_0$ and its highest-degree neighbor, the next node $u_1$ along the path to the root must satisfy $u_1 < e^2 \cdot |C_{m+1}|$.

        In general, for any $u_{t-1}$, the same argument can be applied. Suppose that $u_{t-1} > |C_{m'+1}|$ and that $m'$ is the highest value such that this expression is satisfied. By the same argument as above, $u_{t} < e^2 \cdot |C_{m' + 1}|$. For example, if $u_0 > |C_{m+1}|$ (and that $m$ is the highest value such that this expression is satisfied), then $u_1 < e^2 \cdot |C_{m+1}|$, $u_2 < e^2 \cdot |C_{m+2}|$, and $u_3 < e^2 \cdot |C_{m+3}|$.

        Applying this repeatedly, if $u_0 > |C_{m+1}|$ (and that $m$ is the highest value such that this expression is satisfied), then  $$u_M \leq e^2 \cdot |C_{m + M}| \leq e^2 \cdot |C_{M}|.$$
        
        By definition of $M$ (see Definition \ref{def:M}), $|C_M| = 1$ and therefore $u_M \leq e^2$. Since $u_M$ is an integer time, this means that $u_M \leq 7$. This implies that starting from any $u$, in $M$ (which is less than $\ln(n)$) steps, we can reach a node added in the first seven steps of the uniform attachment process following the edges kept by the LCA. By Lemma \ref{lemma:UA-connected-to-centers}, this node will have an edge to the root node. Additionally, from Lemma \ref{lemma:UA-degree-concentration-bounds}, 
        the root node and each of the first seven nodes will have degree $> \lambda_7  + \mu/2 = \mu \cdot (H_{n-1} - H_6) + \mu/2$, and therefore the edge to the root is kept by the LCA.

        Therefore, the spanner that the LCA gives access to keeps a path of length $O(\log n)$ between any two nodes in the graph, specifically a path going through the first node added in the graph.
    \\\\        
    \noindent \textit{Proof of a sparsity of $n + c$ edges.} By \Cref{lemma:UA-degree-concentration-bounds}, the number of nodes with degree $> \lambda_7 + \mu/2 = \mu \cdot (H_{n-1} - H_6) + \mu/2$ is at most 55. The algorithm keeps one edge between each of these nodes. We can attribute every other kept edge to a distinct node; that is, the LCA keeps the edge from every other node to its highest-degree neighbor. Put together, this implies that the number of edges kept is $n + c$, for some constant $c$ independent of $n$.
    \\\\
    \noindent \textit{Proof of local work.} 
    By construction of the LCA, on input $(u, v)$, if the degrees of $u$ and $v$ are both at least $\mu \cdot (H_{n-1} - H_6) + \mu/2$, the edge is kept. Otherwise, adjacency queries are performed $u$ and $v$, which each have degrees at most $O(\mu \log n)$,
    yielding the stated worst-case time complexity. 
    Additionally, in this case, the average-case (over all possible queries) time complexity is $O(\mu)$; identifying nodes by their time of arrival and letting $d_t$ be the degree of the node that arrived at time $t$, the average time complexity is bounded above by:
    \[\frac{2}{|E|}\sum_{t} d_t^2 = \frac{2}{(n-1) \mu} \sum_{t = 1}^n O\left(\mu \ln\left( \frac{n}{t}\right)\right)^2 = O(\mu).\qedhere\]
\end{proof}

\paragraph{A note on trade-offs between local work and sparsity.} The amount of local work (in the worst case) can be reduced with small increases to the number of edges in the spanner by choosing a different degree threshold in the algorithm. As described above, on an input $(u, v)$, the algorithm checks if both $u$ and $v$ have degrees greater than $D := \mu \cdot (H_{n - 1} - H_6) + \mu/2$, and keeps the edge if this is the case. Otherwise, the edge is kept only when $u$ is the highest-degree neighbor of $v$ or vice versa. This degree $D$ is chosen to correspond to the degree threshold that the first $e^2 \cdot |C_M|$ nodes' degrees will be above. The sparsity corresponded to: $n - e^2 \cdot |C_M| + \left( e^2 \cdot |C_M| \right)^2$, and the amount of local work corresponded to this threshold $D$, in the sense that when $\min \{d_u, d_v\} > D$ the local work is $O(1)$ and otherwise the local work is $O(d_u + d_v)$ where at least one of $d_u$ and $d_v$ is at most $D$. However, one may choose to move the threshold to correspond to a different intermediate center $C_m$ used in the analysis of the algorithm. The algorithm could keep edges whose adjacent vertices both have degrees greater than $\mu \cdot (H_{n-1} - H_{e^2 \cdot |C_m| - 1}) + \mu/2$ and otherwise perform the same procedure as before for choosing whether to keep an edge. The number of edges in the spanner will slightly increase and the local work will slightly decrease.

\paragraph{A note on connections to root-finding.} As noted in Section \ref{app:PA}, for preferential attachment graphs, a connection can be made between local information algorithms for root-finding and local computation algorithms for spanners. This connection can also be made in the setting of uniform attachment graphs with sufficiently high degree parameter $\mu > c_{\mu} \log(n)^2$, for which local root-finding algorithms have not previously been studied. Particularly, the spanner LCA and its analysis give rise to a local root-finding algorithm for uniform attachment graphs. Consider the local algorithm that, starting on any input vertex $u$, follows the path of highest-degree neighbors (meaning the path with $u$, the highest-degree neighbor $v_1$ of $u$, the highest-degree neighbor $v_2$ of $v_1$, and so forth) until it reaches a node $w$ of degree at least $\mu \cdot H_{n-1} - \frac{\mu}{2}$. Return $w$ and all neighbors $w'$ of $w$ such that $\text{deg}(w') > \mu \cdot H_{n-1} - \frac{\mu}{2}$.

Using the same analysis as in the proof of \Cref{thm:UA_Spanner}, we find that the root is in this set with high probability. Moreover, this set is of constant size. This algorithm takes $O(\log n)$ time when we assume that the highest-degree neighbor of a node can be found in $O(1)$ time, which is a standard assumption for local root-finding algorithms \cite{PA1, PA2, PA3}. The time corresponds to the number of nodes explored by the algorithm.

%% file: 900-joint-sampling.tex
\section{Joint Sampling}\label{app:joint-samp}

Finally, we describe our algorithm that provides access to a random graph together with its MIS. 
\ERMIS*

\begin{proof}
    Our algorithms works as follows.
    \paragraph{Global Implementation.} We first describe a global procedure that samples $G\la \Gnp$, in a way that we can later modify to have our desired locality property. Observe that to sample $G$, we can choose an arbitrary order to determine the status of edges $(i,j)$, and in fact this order can be adaptive, as long as each edge is independent.
    
    We initialize a sorted list $M_1=(1)$ and sequentially determine the status of edges $(1,2),(1,3),\ldots$, where each edge is retained in the graph with probability $p$. 
    We halt on the first $i$ such that $(1,i)\notin G$, at which point we set $M_2=M_1\circ (i)$.
    Next, we sample edges from $M_2$ to $i+1,\ldots,$ until we determine the first $i'$ such that $(v,i')\notin G$ for all $v\in M_2$, 
    upon which we again set $M_3=M_2\circ (i')$. 
    We continue in this fashion until the counter $i$ reaches $n$, and set $M$ equal to the final $M_j$.
    After this, we sample all remaining edges independently in an arbitrary order. Observe that this process is clearly equivalent to sampling $G\la \Gnp$,
    and moreover $M$ is an MIS, as every vertex $v$ is connected to an element of $M$ (and in fact is connected to an element with index less than $v$).

    \paragraph{Local Implementation.}We now modify the sampling procedure while keeping the ultimate distribution unchanged. Divide the sampling process into phases $P_j$,
        where in phase $P_j$ we sample edges from $M_j=(v_1,\ldots,v_j)$. 
        Let $K$ be the random variable of the index of the next vertex that is not connected to every $v\in M_j$. 
        This index is distributed $K\sim \Geom\left((1-p)^j\right)+v_j$. 
        For every possible configuration of edges in $M_j \times \{v_j+1,\ldots,n\}$, conditioning on the value of $k=K$ is equivalent to conditioning on the event
        \[
        \text{for all $s\in [j]$, }(v_s,k)\notin G \bigwedge  \text{ for all $v_j<a<k$, there exists $t\in [j]$ such that }(v_t,a)\in G
        \]
        Furthermore, observe that for every vertex $a$ where $a<v_j+k$, we can sample from the conditional edge distribution $(a,M_j)$ conditioned on the value of $K$
        \[
            E_a=(v_1,a),\ldots,(v_j,a) | k=K
        \]
        by sampling each edge in $(a,M)$ independently with probability $p$ and, if no edge is present, rejecting and retrying. It is this procedure that we will use in our algorithm. 
        
        On every query $(u,v)$, our algorithm first determines $M=(v_1,\ldots,v_t)$ by repeated sampling from $\Geom$ with the correct parameters. Subsequently, our connectivity rule is as follows. Given $(u,v)$:
        \begin{enumerate}
            \item If $u\notin M$ and $v\notin M$, we add $(u,v)$ to $G$ independently with probability $p$.
            \item If $u\in M$ and $v\in M$, we do not add $(u,v)$ to $G$.
            \item If $v\in M$ and $u\notin M$, let $v_i<u<v_{i+1}$ be elements of $M$ that bracket $u$. By definition of the global sampling rule, $v\cap \Gamma((v_1,\ldots,v_l))$ is nonempty, and as in the global sampling procedure we sample $(v_1,u),\ldots,(v_i,u)$ independently with probability $p$, and reject and retry if no edges are retained, and once we sample a nonzero neighborhood determine the edges in this fashion. 
            Finally, for $(v_b,u)$ for $b>i$, we again sample this edge independently with probability $p$.

             In order to provide a consistent view of $(G,M)$, $\cA$ designates a fixed section of random tape to be used for generating the MIS, and for all other sampling procedures in the algorithm. To determine $v_{i+1}$ from $M=(v_1,\ldots,v_i)$, we draw from $\Geom$ using~\Cref{lem:GeomSamp}.
        \end{enumerate} 
    \paragraph{Query Time for LCA.} There are two primary components of the runtime, both of which can be bounded in terms of the ultimate size of $M$.
        \begin{claim}
            The final size of $M$ is at most $O(\log(n)/p)$ with high probability.
        \end{claim}
        \begin{proof}
            We have that $|M|$ is bounded by the size of the maximum independent set in $G$, which is itself the size of the maximum clique in the dual graph. As the dual graph is distributed $G(n,q=1-p)$, we appeal to the well-known result~\cite{maxclique} that the maximum clique in a random graph has size $O(\log(n)/\log(1/q))$ with high probability.
            \begin{align}
                \log(1/q) = \log\left(\frac{1}{1-p}\right)
                 = \log\left(1 + \frac{p}{1-p}\right)
                 \ge \frac{\frac{p}{1-p}}{1 + \frac{p}{1-p}} = p
            \end{align}
            This gives us the bound of $O(\log(n)/p)$
        \end{proof}
        
        We then note that the runtime is dominated first by determining $M$, which we do using $|M|$ calls to~\Cref{lem:GeomSamp}, and hence takes total time $|M|\cdot \polylog(n)=\polylog(n)/p$ with high probability. Second, to determine if $e\in G$ and the edge falls into the third case, we perform rejection sampling where our success probability is at least $p$ in each iteration, and hence we terminate after $O(\log(n)/p)$ iterations with high probability and hence the total work is again bounded as $\polylog(n)/p$.
\end{proof}